\documentclass[Crown,sagev,times]{sagej}

\usepackage{moreverb,url}
\newcommand\BibTeX{{\rmfamily B\kern-.05em \textsc{i\kern-.025em b}\kern-.08em
T\kern-.1667em\lower.7ex\hbox{E}\kern-.125emX}}
\def\volumeyear{2017}

\usepackage{etex}
\reserveinserts{18}

\usepackage{amsmath,amssymb,amsfonts,amsthm}
\usepackage{amsxtra,amstext,latexsym,dsfont} 
\usepackage{mathtools}
\usepackage{wasysym}
\usepackage{stmaryrd}
\usepackage{yfonts}
\usepackage{parallel,accents}

\normalfont
\usepackage{bm}

\usepackage{pgf,tikz}
\usetikzlibrary{shapes}
\usetikzlibrary{snakes}
\usetikzlibrary{backgrounds}
\usetikzlibrary{arrows}
\usetikzlibrary{patterns}
\usetikzlibrary{fit}
\usetikzlibrary{calc}
\usetikzlibrary{matrix}

\usepackage{lscape} 

\usepackage{array,longtable,booktabs,float,rotating}
\usepackage{dcolumn,ctable}
\usepackage{enumitem}

\usepackage[flushleft]{threeparttable}

\usepackage{ragged2e}

\usepackage{multirow,multicol}

\usepackage{fancyhdr,extramarks}

 \makeatletter
 \newcommand\sub{\@startsection%
     {subsubsection}{5}{0mm}{-1\baselineskip}{.01\baselineskip}%
     {\normalfont\itshape}}
 \renewcommand\subsubsection{\@startsection%
     {subsubsection}{3}{0mm}{-1\baselineskip}{.01\baselineskip}%
     {\normalfont\itshape}}
 \makeatother

\makeatletter
        \newcommand\Appendix[2][?]{%
            \refstepcounter{section}%
            \addcontentsline{toc}{appendix}%
                {\protect\numberline{\appendixname~\thesection}#1}%
            {\raggedleft\bfseries \appendixname\
                \thesection\par \centering#2\par}%
                \sectionmark{#1}%
                \@afterheading
                \addvspace{\baselineskip}}
        \newcommand\sAppendix[1]{%
            \raggedleft\bfseries\appendixname\par
            \@afterheading\addvspace{\baselineskip}}
    \makeatother

\normalsize \makeindex

\newcolumntype{A}{>{\centering}p{100pt}}
\newlength\savedwidth

\def\coldot{.}%
{\catcode`\.=\active%
    \gdef.{$\egroup\setbox2=\hbox to \dimen0 \bgroup$\coldot}}
\def\rightdots#1{%
    \setbox0=\hbox{$1$}\dimen0=#1\wd0%
    \setbox0=\hbox{$\coldot$}\advance\dimen0 \wd0%
    \setbox2=\hbox to \dimen0 {}%
    \setbox0=\hbox\bgroup\mathcode`\.="8000 $}
\def\endrightdots{$\hfil\egroup\box0\box2}

\newcolumntype{d}[1]{D{.}{.}{#1}}
\newcolumntype{A}{>{\centering}p{100pt}}
\newcolumntype{.}{D{.}{.}{-1}}
\newcolumntype{P}[1]{>{\centering\arraybackslash}p{#1}}
\newcolumntype{M}[1]{>{\centering\arraybackslash}m{#1}}

\DeclareFontFamily{U}{euc}{}
\DeclareFontShape{U}{euc}{m}{n}{<-6>eurm5<6-8>eurm7<8->eurm10}{}%

\theoremstyle{plain}      
\theoremstyle{plain}      
\theoremstyle{plain}      
\theoremstyle{plain}      
\theoremstyle{definition} 
\theoremstyle{definition} 
\theoremstyle{definition} 
\theoremstyle{plain} 
\theoremstyle{definition} 
\theoremstyle{plain} \newtheorem{pro}{Proposition}
\theoremstyle{definition} 
\theoremstyle{definition} 
\theoremstyle{definition} 

\newcounter{nctr}
\newenvironment{3table}{\begin{threeparttable}}{\end{threeparttable}}

\newenvironment{en}{\begin{enumerate}}{\end{enumerate}}

\usepackage[rightbars,color]{changebar}
\cbcolor{blue}

\usepackage{url}
\usepackage{fancyvrb}

\newcommand\mcol{\multicolumn}

\newcommand\bb{\mathbb}

\newcommand\te{\text}
\newcommand\ma[1]{\te{\bf{#1}}}

\newcommand\op{\operatorname}

\newcommand\argmin{\operatornamewithlimits{argmin}}

\newcommand\bias{\operatorname{\bb{B}ias}}
\newcommand\cov{\operatorname{\bb{C}ov}}

\newcommand\df{\te{df}}
\newcommand\E{\bb{E}}

\newcommand\iid{\op{iid}}

\newcommand\mse{\op{MSE}}

\newcommand\lan{\langle}
\newcommand\ov{\overline}
\newcommand\p{\bb{P}}           

\newcommand\pri{^\prime}

\newcommand\rank{\te{rank}}

\newcommand\ran{\rangle}

\newcommand\stack{\stackrel} 

\newcommand\tr{\op{tr}}

\newcommand\tth{^\text{th}}
\newcommand\var{\operatorname{\bb{V}ar}}

\newcommand\wh{\widehat}
\newcommand\wti{\widetilde}

\newcommand\R{\bb{R}}  

\newcommand\bern{\op{Bern}}

\newcommand\by{\ma{y}}

\newcommand\bP{\ma{P}} 
\newcommand\bV{\ma{V}} 

\newcommand\bZ{\ma{Z}}

\newcommand\bzero{\bm{0}} 

\newcommand\al{\alpha}
\newcommand\be{\beta}
\newcommand\ga{\gamma}
\newcommand\de{\delta}
\newcommand\ep{\varepsilon}

\newcommand\ka{\kappa}
\newcommand\kap{\kappa}

\newcommand\sig{\sigma}

\newcommand\bbe{\bm\beta}
\newcommand\bga{\bm\gamma}

\newcommand\bth{\bm\theta}

\usepackage{etoolbox}

\makeatletter
\providerobustcmd*{\bigcapdot}{%
  \mathop{%
    \mathpalette\bigop@dot\bigcap
  }%
}
\newrobustcmd*{\bigop@dot}[2]{%
  \setbox0=\hbox{$\bigcup\m@th#1#2$}%
  \vbox{%
    \lineskiplimit=\maxdimen
    \lineskip=-0.6\dimexpr\ht0+\dp0\relax
    \ialign{%
      \hfil##\hfil\cr
      $\m@th\centerdot$\cr
      \box0\cr
    }%
  }%
}
\makeatother

\makeatletter
\providerobustcmd*{\cupdot}{%
  \mathop{%
    \mathpalette\op@dot\cup
  }%
}
\providerobustcmd*{\capdot}{%
  \mathop{%
    \mathpalette\op@dot\cap
  }%
}
\newrobustcmd*{\op@dot}[2]{%
  \setbox0=\hbox{$\m@th#1#2$}%
  \vbox{%
    \lineskiplimit=\maxdimen
    \lineskip=-0.9\dimexpr\ht0+\dp0\relax
    \ialign{%
      \hfil##\hfil\cr
      $\m@th\cdot$\cr
      \box0\cr
    }%
  }%
}
\makeatother

\makeatletter
\providerobustcmd*{\subsetdot}{%
  \mathop{%
    \mathpalette\subop@dot\subset
  }%
}
\providerobustcmd*{\supsetdot}{%
  \mathop{%
    \mathpalette\subop@dot\supset
  }%
}
\newrobustcmd*{\subop@dot}[2]{%
  \setbox0=\hbox{$\m@th#1#2$}%
  \vbox{%
    \lineskiplimit=\maxdimen
    \lineskip=-0.9\dimexpr\ht0+\dp0\relax
    \ialign{%
      \hfil##\hfil\cr
      $\m@th\cdot$\cr
      \box0\cr
    }%
  }%
}
\makeatother

\makeatletter
\providerobustcmd*{\subseteqdot}{%
  \mathop{%
    \mathpalette\subeqop@dot\subseteq
  }%
}
\providerobustcmd*{\supseteqdot}{%
  \mathop{%
    \mathpalette\subeqop@dot\supseteq
  }%
}
\newrobustcmd*{\subeqop@dot}[2]{%
  \setbox0=\hbox{$\m@th#1#2$}%
  \vbox{%
    \lineskiplimit=\maxdimen
    \lineskip=-0.7\dimexpr\ht0+\dp0\relax
    \ialign{%
      \hfil##\hfil\cr
      $\m@th\cdot$\cr
      \box0\cr
    }%
  }%
}
\makeatother

\newcommand*\Otri{\ensuremath{
        \;
        \begin{tikzpicture}
          \draw[thick] (0pt,0pt) circle (6.5pt);
          \draw[thick] (-6pt,+0pt) -- (+5.50pt,-2.75pt);
          \draw[thick] (-6pt,+0pt) -- (+5.50pt,+2.75pt);
        \end{tikzpicture}
        \;
        }
}

\newcommand{\dobigtri}[1]{%
  \vcenter{#1\kern.2ex\hbox{$\Otri$}\kern.2ex}}

\newcommand{\cse}{\op{CSE}}

\newcommand{\Cor}{\op{Cor}}
\newcommand{\nde}{\op{NDE}}
\newcommand{\nie}{\op{NIE}}
\newcommand{\tde}{\op{TE}}

\newcommand\aols{\wti} 
\newcommand\atsls{\wh} 

\begin{document}
\sloppy
\runninghead{Ginestet, Emsley, and Landau}
\title{Stein-like Estimators for Causal Mediation Analysis \\ 
       in Randomized Trials}

\author{Cedric E.~Ginestet\affilnum{1}, Richard Emsley\affilnum{2,3},
        and Sabine Landau\affilnum{1}}
 
\affiliation{
\affilnum{1}Department of Biostatistics and Health Informatics,
Institute of Psychiatry, Psychology and Neuroscience, King's College
London, \affilnum{2} MAHSC Clinical Trials Unit, The University of
Manchester, Manchester Academic Health Science Centre,
\affilnum{3} Centre for Biostatistics, School of Health Sciences, The
University of Manchester, Manchester Academic Health Science Centre}
\corrauth{Cedric E.~Ginestet
Department of Biostatistics and Health Informatics
Institute of Psychiatry, Psychology and Neuroscience
King's College London, PO20, 16 De Crespigny Park,
London SE5 8AF, UK}
\email{cedric.ginestet@kcl.ac.uk}

\begin{abstract}
Causal mediation analysis aims to estimate the natural direct and
indirect effects under clearly specified assumptions. Traditional
mediation analysis based on Ordinary Least Squares (OLS) relies on the
absence of unmeasured causes of the putative mediator and outcome. When
this assumption cannot be justified, Instrumental Variables (IV)
estimators can be used in order to produce an asymptotically
unbiased estimator of the mediator-outcome link. However, provided
that valid instruments exist, bias removal
comes at the cost of variance inflation for standard IV procedures such
as Two-Stage Least Squares (TSLS). A Semi-Parametric Stein-Like (SPSL)
estimator has been proposed in the literature that strikes a natural
trade-off between the unbiasedness of the TSLS procedure and the
relatively small variance of the OLS estimator. Moreover, the SPSL has
the advantage that its shrinkage
parameter can be directly estimated from the data. In this paper,
we demonstrate how this Stein-like estimator can be implemented in
the context of the estimation of natural
direct and natural indirect effects of treatments in randomized controlled trials.
The performance of the competing methods is studied in a simulation
study, in which both the strength of hidden confounding and the strength of the 
instruments are independently varied. These considerations are motivated by a trial
in mental health evaluating the impact of a primary care-based
intervention to reduce depression in the elderly.  
\end{abstract}
\maketitle

\section{Introduction}\label{sec:introduction}
Mediation analysis has become a popular approach to data
analysis in a variety of disciplines. This approach permits to
study alternative causal paths linking an experimental factor of
interest with a particular outcome\citep{MacKinnon2008}. It
has been especially successful in the context of
mental health, where psychologists and psychiatrists are particularly
interested in the mechanisms of action of a given treatment. These
mechanisms are usually studied with respect to certain intermediate variables
that are likely to be related to the personality, cognition and social
environment of the individuals that are taking part in the study. 

In mental health, we are often concerned with evaluating the
effect of psychological therapy on clinical outcome, with respect to certain
intermediate variables. When the indirect effect of the treatment
through the intermediate variable is of interest, such a variable is
referred to as a \emph{target mediator}. By contrast, when we are
controlling for the intermediate variable, and the primary interest of
the study lies in estimating the direct effect of treatment on the
outcome; we refer to such a variable as a \emph{nuisance
  mediator}. Often, the distinction between a target and a nuisance
mediator depends on whether or not the mediator constitutes an
alternative form of treatment. This is the case in the PROSPECT data
set that motivates this study, in which the effect of
psychotherapy is mediated by adherence to a course of anti-depressant
medication. 

Several theoretical frameworks have been proposed for studying
mediation from a causal perspective. Such approaches tend to build
upon the foundational work of Baron and Kenny\citep{Baron1986} (1986),
who have established the basis of mediation analysis. 
This framework has then been formalized in order to allow for causal
inference. The first formalization of causal mediation analysis was
given by Robins and Greenland\citep{Robins1992} (1992); and several
variants have been proposed in the
literature, including the works of Pearl\citep{Pearl2001} (2001),
Rubin\citep{Rubin2004a} (2004), and VanderWeele\citep{Vanderweele2008}
(2008). In the paper at hand, we will describe causal mediation in terms of
potential outcomes, using the notation and the set of assumptions
adopted by Imai, Keele and Yamamoto\citep{Imai2010}
(2010). Throughout this article, we will assume that the outcome of
interest is continuous. 
In this setting, the main estimands of interest are the natural direct and
natural indirect effects, denoted NDE and NIE respectively. In trials,
such quantities can be estimated without bias, under the assumption that 
the intermediate variable is exogenous in the model for the
outcome. (A predictor of the outcome
variable is said to be \emph{exogenous}, whenever it is not correlated
with the error term in the model, and \emph{endogenous}, otherwise.)

In practice however, this exogeneity assumption can be difficult to
justify, due to the likely presence of baseline variables that are
common causes of the intermediate and the clinical outcome
variable. One of the proposed solutions to this problem
has been the use of instrumental variables (IVs), which can be combined
with mediation analysis, in order to draw causal
inference. (See Lynch et al.\citep{Lynch2008} (2008) and
Ten Have et al.\citep{TenHave2012} (2012) for a review of causal
mediation analysis.) In this article, we will specifically focus on
the use of interaction terms as instruments, constructed by
interacting the experimental factor with the baseline
covariates. Note, however, that our methods can readily be generalized
to other IVs. 

The most common estimator using IVs is the Two-Stage Least
Squares (TSLS)\citep{Wooldridge2002}, which relies on further
assumptions about the behavior of the candidate instruments. Under
these additional assumptions, the asymptotic properties of the TSLS
estimator are well-understood. Provided that the instruments solely
affect the outcome through the endogenous variable of interest, the
TSLS estimator is guaranteed to be asymptotically
unbiased\citep{Wooldridge2002}. For finite sample sizes, however, the
decrease in bias associated with the use of this estimator, will 
lead to an increase in variance. In particular, there may be
situations in which the variance increase of the TSLS estimator does
not warrant preferring that estimator over the potentially biased Ordinary
Least Squares (OLS) estimator. 

In this paper, we follow the lead of Judge and Mittelhammer
\citep{Judge2004} (2004), who have constructed a combined estimator, which
strikes a trade-off between the OLS and the TSLS estimators, by
minimizing the Mean Squared Error (MSE) of the resulting combined
estimator. This method closely resembles the so-called Stein estimator,
originally introduced by James and Stein\citep{James1961} (1961), and
made popular by Efron\citep{Efron1973} (1973). Stein estimators
have anticipated some of the central ideas of Bayesian
statistics, by shifting the main focus of statistical analysis from
minimizing an estimator's unbiasedness to minimizing an estimator's
MSE. These ideas are best articulated within the language
of decision theory (See Berger\citep{Berger1985} (1985), for an
introduction to decision theory). From this perspective, the MSE
can be formalized as a loss function, and the optimal estimator is the
one that minimizes that quantity. The Semi-Parametric Stein-Like
(SPSL) estimator is defined as an affine combination of the OLS
and TSLS estimators; where the shrinkage parameter controlling
the respective contributions of the OLS and TSLS estimators can be 
estimated from the data, under the assumption that the TSLS estimator
is asymptotically unbiased. 

The main contributions of this paper are twofold. Firstly, we provide the first use
of the SPSL estimator in the context of
causal mediation analysis. The SPSL will here be compared with
standard estimators, including the OLS and TSLS
estimators for estimating the effect of endogenous intermediate
variables in causal mediation; where OLS estimation here corresponds
to the standard Baron--Kenny approach. Note that the Baron--Kenny
framework generally assumes that the mediators are continuous, whereas our
approach enables us to accommodate binary mediators. The asymptotic behaviors of the
family of SPSL estimators have recently been well-studied
\citep{Mittelhammer2005,Judge2012,Judge2012a,Judge2013}.
The SPSL estimator has been used to investigate Local Average
Treatment Effects (LATEs) in dose-response models\citep{Ginestet2017a}.
However, to the best of the authors' knowledge, this family of
estimator has not been used in the context of causal mediation analysis, when
the estimation of the causal path from the intermediate variable
to the outcome is potentially biased, due to unmeasured confounding. 

Secondly, we generalize the SPSL estimator by allowing for the 
selection of a subset of parameters that affects the optimization of the
shrinkage parameter. Indeed, in many circumstances, one is solely
interested in the estimation of a particular set of estimands, and it
is therefore convenient to be able to restrict the dependence of the
shrinkage parameter on the MSE of a subset of target estimands. In
this paper, we implement such a restriction by introducing a
projection matrix, which permits to restrain the estimation of the
shrinkage parameter to a subset of the parameters of interest, such as
the direct effect of treatment, for instance. 

Our use of the SPSL estimator for causal mediation analysis is motivated by 
a clinical trial in mental health. The Prevention of
Suicide in Primary Care Elderly: Collaborative Trial, more
concisely referred to as PROSPECT\citep{Bruce2004}, is a randomized controlled trial,
which tested the effect of a primary care intervention on major risk
factors for suicide in an elderly population, and in which the
intermediate variable is whether or not patients are taking antidepressant
medication. This particular study
has served as a motivating example for several causal mediation
analyses previously published in the literature, including studies by
Ten Have et al.~(2007)\citep{TenHave2007}, 
Emsley et al.~(2010)\citep{Emsley2010}, and Small
(2012)\citep{Small2012}. In the paper at hand, we replicate some of
these previous results, and compare them with the performance of the
SPSL estimator for this data set. 

The PROSPECT data set is an unusual example of a mediation analysis,
since the main estimand of interest is the NDE.
That is, we wish to evaluate whether or not the
psychotherapeutic intervention affects the outcome, after having controlled
for the effect of taking antidepressant medication. This should be
contrasted with most other mediation studies, in which one is typically
interested in estimating the NIE --that is, the effect of the
target mediator on the outcome. In contradistinction, the
intermediate variable in the PROSPECT data may be regarded as a
nuisance mediator, which is solely of secondary interest to the
trialists.

The paper is organized as follows. In the first section, we
introduce the causal estimands of interest, and describe how such
parameters can be estimated using the OLS, the TSLS and the
SPSL estimator. The performance of our three competing estimators for
causal mediation analysis is then evaluated by means of a
Monte Carlo simulation study, in the second section. The methods are then
applied to a re-analysis of the PROSPECT data set, in the third section;
and we close with a discussion of the limitations and further
generalizations of such estimators in our final section. The proofs of
the main results in the paper are deferred to an appendix. 

\section{Causal Mediation Analysis}\label{sec:methods}
\subsection{Causal Estimands}\label{sec:estimand}
The sample data are assumed to have been collected as part of a
clinical trial, in which $R_{i}$ denotes randomized treatment offer to
the $i\tth$ subject. The clinical outcome of interest, denoted by $Y_{i}$, is
a continuous post-randomization variable, and $M_{i}$ is the putative
mediator under investigation, which is also a post-randomization
variable. The mediator may be either binary or continuous. 
In addition, there are also $k$ pre-randomization (or baseline)
variables, denoted by a random vector, $X_{i}$, such that
$X_{i}:=(X_{i1},\ldots,X_{ik})\pri$. Without any loss of generality,
these baseline variables may also be either binary or continuous.

For every $r\in\{0,1\}$, and for every $m\in\R$, the potential outcome
$Y_{i}(r,m)$ is defined as the outcome that would be observed for the
$i\tth$ subject, if $R_{i}$ and $M_{i}$ were to take values $r$ and
$m$, respectively. Several possible mechanisms have
been proposed in the literature that allow the potential outcomes,
$Y_{i}(r,m)$, to take different values according to different choices
of $r$ and $m$ \citep{TenHave2007,Small2012}.
Similarly, the potential mediator, $M_{i}(r)$, is defined as the
value taken by the mediator in the $i\tth$ subject, when the value of
$R$ is $r$. The aforementioned observed outcomes and observed
mediators are then defined as a function of the potential outcomes and
potential mediators, such that we have $Y_{i}:=Y_{i}(R_{i},M_{i})$,
and $M_{i}:=M_{i}(R_{i})$, respectively.

For every subject, every $r$, and every $m$, the potential outcomes
are given the following structural model,
\begin{equation}\label{eq:model potential}
     Y_{i}(r,m) := Y_{i}(0,0) + \be_{R,i}r + \be_{M,i}m,
\end{equation}
with $Y_{i}(0,0) := \bbe_{X}\pri X_{i} + \omega_{i}$, and $\E[\omega_{i}]=0$.
and where the parameters, $\be_{M,i}$ and $\be_{R,i}$, can vary
between subjects reflecting treatment effect and
mediator effect heterogeneity, respectively.

The parameters in model \eqref{eq:model potential} can thus be interpreted
in the following manner. Given a
subject $i$, the parameter, $\be_{M,i}$, denotes the effect caused by a
unit increase in the mediator on the outcome, holding treatment level at
$r$. Similarly, $\be_{R,i}$ should be interpreted as the effect of
treatment on the outcome, while holding the mediator constant at 
level $m$. Finally, we will respectively denote by $\be_{M}:=\E[\be_{M,i}]$ and 
$\be_{R}:=\E[\be_{R,i}]$, the average causal effect of the mediator and the
average effect of the treatment on the outcome.

Using our definitions of the observed outcome, $Y_{i}$, and of the
counterfactual, $Y_{i}(0,0)$; we obtain the following linear model for
the observed outcome, 
\begin{equation}\label{eq:model}
      Y_{i} = \bbe_{X}\pri X_{i} + \be_{R}R_{i} + \be_{M}M_{i} + \ep_{i}, 
\end{equation}
for every $i=1,\ldots,n$; in which $\bbe_{X}$ is a $k$-dimensional
column vector of unknown parameters containing an intercept, and where
the error terms comprise the individual deviations from the average
causal effect,
\begin{equation}\label{eq:ep}
    \ep_{i} := 
    \omega_{i}
    + \big(\be_{M,i}-\be_{M}\big)M_{i} 
    + \big(\be_{R,i}-\be_{R}\big)R_{i}.
\end{equation}
where recall that $\omega_{i}=Y_{i}(0,0) - \E[Y_{i}(0,0)|X_{i}]$.

Treating the mediator as unobserved, the potential outcomes can be
described by the following structural model,
\begin{equation}\notag
      Y_{i}(r) = \bth_{X}\pri X_{i} + \theta_{R,i}R_{i} + \xi_{i},
\end{equation}
where as before, $\theta_{R}:=\E[\theta_{R,i}]$ denotes the average
effect of treatment offer on the outcome, and with
$\E[\xi_{i}]=0$. This then leads to the following model for the
observed outcomes, 
\begin{equation}\label{eq:total}
      Y_{i} = \bth_{X}\pri X_{i} + \theta_{R}R_{i} + \nu_{i},
      \qquad\te{with }\nu_{i}:=\xi_{i}+(\theta_{R,i}-\theta_{R}). 
\end{equation}
This model is represented in Figure \ref{fig:total}. 

For continuous mediators, we can translate our choice of
notation, into the conventional Baron-Kenny notation
\citep{Baron1986}. If we were to represent the average causal effect of
$R$ on $M$ by $\ga_{R}$; we could then adopt the following notation,
$a:=\ga_{R}$, $b:=\be_{M}$, $c\pri:=\be_{R}$, and $c:=\theta_{R}$.
\begin{figure}[t]
\centering
\begin{tikzpicture}
    \draw (-2.0,1) node[draw,inner sep=8pt](r){$R$};
    \draw (-2.0,-1) node[draw,inner sep=8pt](x){$X$};
    \draw (+2.0,0) node[draw,inner sep=8pt](y){$Y$};
    \draw (2.0+1.5,0) node[draw,circle,inner sep=6pt](e){};
    \draw[->] (r) -- (y) node[midway,anchor=south]{$\theta_{R}$};
    \draw[->] (x) -- (y) node[midway,anchor=south]{$\bth_{X}$};
    \draw[->] (e) -- (y) node[midway,anchor=north west]{};
\end{tikzpicture}
\vspace{.25cm}
\caption{Graphical representation of the model of the total effect
  (TE) of $R$ on $Y$, as described in equation (\ref{eq:total});
  where the empty circle denotes an error term. Since
  subjects have been randomly assigned to the different levels of
  treatment allocation, $R$; it follows that $R$ is an exogenous
  predictor of $Y$. \label{fig:total}}
\end{figure}
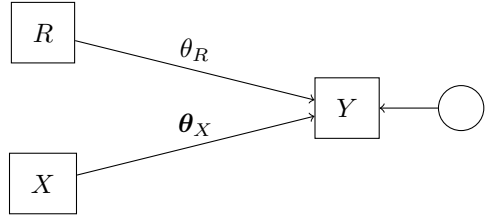

The main estimands of interest will be the natural direct effect (NDE)
and the natural indirect effect (NIE). For continuous outcomes, the
total effect (TE) can be decomposed such that 
\begin{equation}\notag
    \tde := \E\big[Y_{i}(1) - Y_{i}(0)\big] 
          = \E\big[Y_{i}(1,M_{i}(1)) - Y_{i}(0,M_{i}(0))\big].
\end{equation}
In our notation, TE corresponds to effect of treatment offer on the
outcome, according to the structural model in equation \eqref{eq:total}, such that
$\tde=\theta_{R}$. The NDE, on the other hand, is defined as follows,
\begin{equation}\notag
    \nde:= \E[Y_{i}(1,M_{i}(0)) - Y_{i}(0,M_{i}(0))] 
        = \be_{R}.
\end{equation}
Finally, for continuous $Y_{i}$'s, the NIE can be expressed as a
difference between these two estimands. Formally, this gives
\begin{equation}\notag
    \nie := \E[Y_{i}(1,M_{i}(1)) - Y_{i}(1,M_{i}(0))] 
         = \theta_{R}-\be_{R}.
\end{equation}
This expression for the NIE with continuous outcomes is convenient,
because it covers both continuous and binary mediators. 

\subsection{OLS Estimator}\label{sec:ols}
Observe that since both the baseline covariates,
$X_{i}$'s, and the randomization variable, $R_{i}$'s, are exogenous,
it follows that parameters, $\ga_{R}$ and $\theta_{R}$,
can be unbiasedly estimated using OLS. However, there is no guarantee that the
effect of the mediator on the outcome is not confounded by an
unmeasured variable. Therefore, a naive OLS estimator of the
parameter, $\be_{M}$, may be biased. Similarly, the OLS estimator of
$\be_{R}$ may also be biased due to the endogeneity of the mediator. 
A simplified version of such a causal mediation model, in the presence
of a confounder, $U_{i}$'s, has been represented in Figure
\ref{fig:confounded model}. In Figure \ref{fig:confounded model},
$\be_{R}$ and $\be_{M}$ are biased due to unmeasured confounding,
since the intermediate variable, $M_{i}$'s, is endogenous in this
figure. 
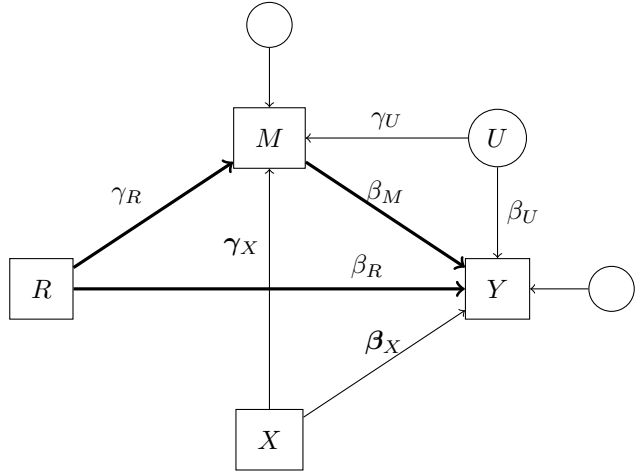
\begin{figure}[t]
\centering
\begin{tikzpicture}
    \draw (-3.0,-2) node[draw,inner sep=8pt](r){$R$};
    \draw (+0,+0) node[draw,inner sep=8pt](m){$M$};
    \draw (+0,-4) node[draw,inner sep=8pt](x){$X$};
    \draw (+3.0,-2) node[draw,inner sep=8pt](y){$Y$};
    \draw (+3.0,+0) node[draw,circle,inner sep=4pt](u){$U$};
    \draw (+0,+1.5) node[draw,circle,inner sep=6pt](d){};
    \draw (3.0+1.5,-2) node[draw,circle,inner sep=6pt](e){};
    \draw[very thick,->] (m) -- (y) node[midway,above]{$\be_{M}$};
    \draw[->] (u) -- (m) node[midway,anchor=south]{$\ga_{U}$};
    \draw[->] (u) -- (y) node[midway,anchor=west]{$\beta_{U}$};
    \draw[very thick,->] (r) -- (m) node[midway,anchor=south east]{$\ga_{R}$};
    \draw[very thick,->] (r) -- (y) node[near end,anchor=south]{$\be_{R}$};
    \draw[->] (x) -- (m) node[near end,anchor=north east]{$\bga_{X}$};
    \draw[->] (x) -- (y) node[midway,anchor=south]{$\bbe_{X}$};
    \draw[->] (e) -- (y) node[midway,anchor=north west]{};
    \draw[->] (d) -- (m) node[midway,anchor=north east]{};
\end{tikzpicture}\\
\vspace{.25cm}
\caption{Graphical representation of the mediation model
   described in equation (\ref{eq:model}) in the presence
   of a confounder, $U$; where empty circles denote an error
   terms. For continuous mediators, we have the
   following correspondence between
   the above notation and the standard
   Baron-Kenny notation: $a=\ga_{R}$, $b=\be_{M}$, and
   $c\pri=\be_{R}$; in which $R$ and $M$ denote treatment offer and
   the mediator, respectively. The three paths of interest in 
   mediation investigations, have been emphasized in bold. 
   \label{fig:confounded model}}
\end{figure}

Various sets of assumptions can be used in order to conduct causal mediation
analysis. For the OLS estimator, we will use a set of assumptions
referred to as sequential ignorability \citep{Imai2010}. For every
$r\in\{0,1\}$, and every $m\in\R$, sequential ignorability assumes that 
\begin{en}[leftmargin=1.75cm]\setlength\itemsep{0em}
  \item[(OLS--1)] $Y_{i}(r,m) \perp R_{i} \,|\, X_{i}$.
  \item[(OLS--2)] $M_{i}(r) \perp R_{i} \,|\, X_{i}$.
  \item[(OLS--3)] $Y_{i}(r,m) \perp M_{i} \,|\, X_{i}$.
\end{en}
These assumptions respectively state the following: 
ignorable treatment assignment in terms of the outcome, given
the covariates, (OLS--1); ignorable treatment assignment in terms of
the mediator, given covariates, (OLS--2); and ignorable mediator
assignment, given covariates (OLS--3).

Observe that, whenever the $R_{i}$'s correspond to random allocation to
treatment offer, as in our motivating trial, it then follows
that conditions (OLS--1) and (OLS--2) are automatically satisfied. The
fact that (OLS--1) holds in our setting, allows us to unbiasedly estimate 
the causal effect of treatment offer on the outcome, denoted by $\theta_{R}$;
using an OLS estimator, denoted by $\wti\theta_{R}$. Similarly, the fact
that (OLS--2) holds permits us to unbiasedly estimate the causal effect of
treatment offer on the mediator, denoted by $\ga_{R}$, using an OLS
estimator, denoted by $\wti\ga_{R}$.

Moreover, it is additionally assumed, for regulatory reasons, that the
following strict inequalities hold, $\p(R=r|X=x)>0$, and
$\p(M=m|R=r,X=x)>0$. Under the model for the potential outcomes described in equation
(\ref{eq:model potential}), the third assumption of sequential
ignorability given in (OLS--3) can be reformulated as follows,
\begin{equation}\notag
       Y_{i}(0,0),\be_{M,i},\be_{R,i} \perp M_{i} \,|\, X_{i}.
\end{equation}
Therefore, under sequential ignorability, the three random variables
on the RHS of equation (\ref{eq:model potential}) are assumed to be
conditionally independent of the mediator, given the values of the
baseline covariates. 

In the absence of unmeasured confounders between the $M_{i}$'s and the
$Y_{i}$'s, sequential ignorability 
holds, and one can estimate the NIE and NDE by computing the OLS
estimator of the direct effect of treatment offer on the outcome, denoted
$\be_{R}$. For convenience, the
parameters of interest in the model described in equation \eqref{eq:model}
will be collectively denoted as a vector, 
\begin{equation}\notag
    \bbe:=(\bbe_{X}\pri,\be_{M},\be_{R})\pri.
\end{equation}
Similarly, all the variables in this model will be expressed as the
random vector, 
\begin{equation}\notag
      V_{i}:= (X_{i}\pri,M_{i},R_{i})\pri,
\end{equation}
where recall that $X_{i}$ represents a $k$-dimensional column vector
of baseline covariates including an intercept, whereas $M_{i}$ and
$R_{i}$ are real-valued random variables, although $M_{i}$ is also
allowed to be binary. Thus, each $V_{i}$ is a 
$(k+2)$-dimensional random vector. In addition, a set of $n$
observations from the $Y_{i}$'s will be denoted by the vector $\by$,
while a set of $n$ realizations from the $V_{i}$'s will take the form
of a matrix of order $n\times(k+2)$, denoted $\bV$.

Under the further assumption that the matrix $\E[V_{i}V_{i}\pri]$ is
full-rank, we can compute the OLS estimator. 
\begin{en}[leftmargin=1.75cm]\setlength\itemsep{0em}
  \item[(OLS--4)] $\rank(\E[V_{i}V_{i}\pri])=k+2$.
\end{en}
The OLS estimator for $\bbe$, which is uniquely given by the vector that
minimizes the empirical MSE, and
takes the form, $\aols\bbe:=(\bV\pri\bV)^{-1}(\bV\pri\by)$. Moreover,
the empirical variance of this estimator is given by
$\wh\sig^{2}(\bV\pri\bV)^{-1}$, where 
$\wh\sig^{2}$ is defined as
$(\by-\bV\aols\bbe)\pri(\by-\bV\aols\bbe)/(n-k-2)$. Moreover, the vector
of parameters for the total effects,
$\bth:=(\bth_{X}\pri,\theta_{R})\pri$, from equation 
(\ref{eq:total}) can also be estimated using OLS, thereby producing the
following estimators of the natural effects: $\wti\nde := \aols\be_{R}$,
and $\wti\nie := \wti\theta_{R}-\aols\be_{R}$. (Hence, observe that
albeit we are estimating the full vector of parameters, $\bbe$; the
sole element of interest in this vector for estimating NDE and NIE is
$\be_{R}$.)

\subsection{TSLS Estimator}\label{sec:tsls}
In the presence (or suspected presence) of unmeasured
confounders, different assumptions are required in order to estimate the
parameters of interest without bias. In the data at hand, although
allocation to treatment has been randomized, both $M_{i}$ and
$Y_{i}$ are post-randomization variables, which may be affected by
common causes. Therefore, one cannot guarantee that the path from
the mediator to the outcome has not been confounded by an unobserved
variable. When unmeasured confounders affect the relationship between the
outcome and the mediator, as illustrated in Figure
\ref{fig:confounded model}, the third portion of sequential
ignorability, (OLS--3) does not hold, and further assumptions are
hence required to ensure that such a model is identifiable. 

Several groups of researchers have used instruments that are defined
as interactions between certain baseline variables and random assignment to treatment 
\citep{Dunn2007,TenHave2007,Albert2008}. Such choices of IVs require a
particular set of assumptions, which ensure that the resulting variables
constitute valid instruments. In this paper, we will consider a variant
of the conditions described by Small \citep{Small2012}. These assumptions apply to general
mediation models that make use of such interaction terms as instruments. 
For consistency with the previous literature on this topic, we will
also adopt some of the notation used by Small \citep{Small2012}, throughout
the rest of this section. However, we should emphasize that SPSL
estimation in causal mediation, is not restricted to the use of
interaction terms as instruments. 
\begin{figure}[t]
\centering
\begin{tikzpicture}[scale=1.2]
    \draw (-3,0) node[draw,inner sep=8pt](z){$RX$};
    \draw (+0,0) node[draw,inner sep=8pt](m){$M$};
    \draw (+4,0) node[draw,inner sep=8pt](y){$Y$};
    \draw (+1,-1.75) node[draw,inner sep=8pt](r){$R$};
    \draw (+3,-1.75) node[draw,inner sep=8pt](x){$X$};
    \draw (+0,+1.5) node[draw,circle,inner sep=6pt](d){};
    \draw (+4,+1.5) node[draw,circle,inner sep=6pt](e){};
    \draw (+2,+1.75) node[draw,circle,inner sep=4pt](u){$U$};
    \draw[->] (z) -- (m) node[midway,above]{$\bga_{RX}$};
    \draw[very thick,->] (m) -- (y) node[midway,above]{$\be_{M}$};
    \draw[very thick,->] (r) -- (m) node[midway,anchor=north east]{$\ga_{R}$};
    \draw[very thick,->] (r) -- (y) node[midway,anchor=south]{$\be_{R}$};
    \draw[->] (x) -- (m) node[midway,anchor=south]{$\bga_{X}$};
    \draw[->] (x) -- (y) node[midway,anchor=north west]{$\bbe_{X}$};
    \draw[->] (u) -- (m) node[midway,anchor=south east]{$\ga_{U}$};
    \draw[->] (u) -- (y) node[midway,anchor=south west]{$\beta_{U}$};
    \draw[->] (e) -- (y) node[midway,anchor=north west]{};
    \draw[->] (d) -- (m) node[midway,anchor=north east]{};
\end{tikzpicture}
\caption{Graphical representation of the instrumented mediation model described in
  equation (\ref{eq:instrumented model}), in which the relationship
  between the mediator, $M$, and the outcome $Y$, is confounded by the
  presence of an unknown variable $U$; where, as before, the empty
  circles denote error terms. The interaction instrument,
  $RX$, is here used to handle the endogeneity of $M$; while 
  the randomization variable, $R$, and the baseline covariates, $X$,
  are all assumed to be exogenous. As in Figure \ref{fig:confounded
    model}, the three links defining the main causal mediation model,
  have been emphasized in bold. 
  \label{fig:instrumented model}}
\end{figure}
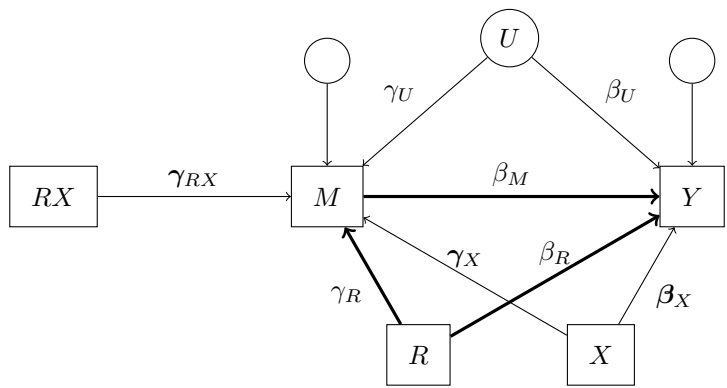

Here, we supplemented the model for the observed outcome, $Y_{i}$'s,
with a predictive model for the observed continuous (or binary) mediator,
$M_{i}$'s, such that we obtain the system of equations that has also
been illustrated graphically in Figure \ref{fig:instrumented model}, 
\begin{equation}\label{eq:instrumented model}
  \begin{aligned}
   Y_{i} &= \bbe_{X}\pri X_{i} + \be_{R}R_{i} + \be_{M}M_{i} + \ep_{i}, \\
   M_{i} &= \bga_{X}\pri X_{i} + \ga_{R}R_{i} + \bga_{RX}\pri R_{i}X_{i} + \de_{i};
  \end{aligned}
\end{equation}
with $\de_{i}:=\E[M_{i}|X_{i},R_{i},]-M_{i}$, and $\E[\de_{i}]=0$; and
where the $\ep_{i}$'s and the $\de_{i}$'s are assumed to be
independent. Furthermore, note that the $M_{i}$'s are here modelled linearly, despite the
fact that this variable may be binary. This does not pose a problem per
se, as long as the instruments, $R_{i}X_{i}$'s, are predictive of the
$M_{i}$'s. That is, mis-specification of the functional form of this
model (e.g.~as a linear regression, when, in reality, this is a logistic
regression), while leading to difficulties interpreting the gamma's; 
does not affect the estimation of the parameters of interest, which
are the beta's in the model for the $Y_{i}$'s, since TSLS estimation
solely requires a correct specification of the model for the
outcomes. Thus, the error terms, $\de_{i}$'s, of the linear model
for the mediator need not be normally distributed.

For convenience, we will define the set of instruments as the following vectors,
\begin{equation}\notag
      Z_{i}:= (X_{i}\pri,R_{i},R_{i}X_{i}\pri)\pri,
\end{equation}
where each such $Z_{i}$ is a $(2k+1)$-dimensional column vector.
Equipped with this notation, we can then state the assumptions
required to guarantee the validity of the $R_{i}X_{i}$'s as
instruments. We will assume that the following conditions hold
for every subject,
\begin{en}[leftmargin=1.75cm]\setlength\itemsep{0em}
  \item[(TSLS--1)] $M_{i}\perp\be_{M,i}|R_{i},X_{i}$.
  \item[(TSLS--2)] $\be_{R}=\E[\be_{R,i}|X_{i}]$, and $\be_{M}=\E[\be_{M,i}|X_{i}]$.
  \item[(TSLS--3)] $\E[Z_{i}Z_{i}\pri]$, and $\E[Z_{i}V_{i}\pri]$ are full-rank.
  \item[(TSLS--4)] $\cov(V_{i},Z_{i})\neq\bzero$.
\end{en}
Here, (TSLS--1) should be interpreted as the independence of the
individual mediator effects with the values taken by the mediator. 
Condition (TSLS--4) is commonly
referred to in the literature on causal inference, as the
\emph{relevance} of the IVs. 
In addition, observe that assumption (TSLS--2) is weaker than the ones made by
previous authors, who have used interaction terms as instruments, and who
have assumed homogeneous treatment
effects\citep{Dunn2007,TenHave2007,Albert2008}, such that the
$\beta_{M,i}$'s, and $\beta_{R,i}$'s are assumed to be identical
for all subjects. Here, by contrast, we have only required these
parameters to have identical conditional expectations conditional 
on the $X_{i}$'s, as stated in condition (TSLS--2).

Also, note that this set of assumptions slightly differs from the one described by Small
\citep{Small2012}, since we have replaced the assumption that this author
refers to as (IV--A1), by an assumption on the ranks of the matrices 
$\E[Z_{i}Z_{i}\pri]$, and $\E[Z_{i}V_{i}\pri]$, which we refer to as
(TSLS--3). The latter assumption is here expressed in terms
of the ranks of the expectations of the cross-products of the vector
of instruments, and the vector of covariates. This condition is a
relatively weak requirement that guarantees the identifiability of the
resulting TSLS estimator \citep{Wooldridge2002}. 

We can now show that the corresponding TSLS estimator weakly converges
to the target vector of the parameters of interest. Firstly, following
Small \citep{Small2012}, we demonstrate that the above assumptions are sufficient
to guarantee the exogeneity of the $R_{i}X_{i}$'s in model
\eqref{eq:instrumented model}. A proof of this proposition has been
relegated to the appendix. 
\begin{pro}\label{pro:uncorrelated}
    Under assumptions (TSLS--1) and (TSLS--2), and under the assumption that the
    $R_{i}$'s are exogenous with respect to the $Y_{i}$'s in model
    \eqref{eq:instrumented model}, we have $\cov(R_{i}X_{i},\ep_{i})=\bzero$.
\end{pro}
In the context of trials, observe that the exogeneity of the $R_{i}$'s
is automatically satisfied. 
It then follows that the TSLS estimator, $\atsls\bbe$, can be computed with respect to
the matrix $\wh\bV$, such that $\atsls\bbe := (\wh\bV\pri\wh\bV)^{-1}(\wh\bV\pri\by)$,
where $\wh\bV$ denotes the projected matrix of the variables in
the second-stage equation with respect to the matrix of instruments,
$\bZ$. Analogously to the OLS, the variance of the estimator is then
given by $\wti\sig^{2}(\bV\pri\bV)^{-1}$, where $\wti\sig^{2}$ is
defined as $(\by-\bV\atsls\bbe)\pri(\by-\bV\atsls\bbe)/(n-k-2)$.
The consistency of the TSLS estimator, can then immediately be derived.
\begin{pro}\label{pro:tsls consistency}
    Under conditions (TSLS--1) to (TSLS--4), and under the assumptions that both the
    $R_{i}$'s and the $X_{i}$'s are exogenous with respect to the $Y_{i}$'s in model
    \eqref{eq:instrumented model}; we have
    $\atsls\bbe\stack{p}{\to}\bbe$. 
\end{pro}
As before, the proof of this proposition is provided in the appendix.
It then suffices to plug in this estimator of $\bbe$ in our
definitions of the natural effects, in order to
construct the TSLS estimators for these causal estimands, such that we
obtain $\wh\nde := \atsls\be_{R}$, and $\wh\nie :=
\wti\theta_{R}-\atsls\be_{R}$; where note that $\wti\theta_{R}$ is still
estimated using OLS, since the randomization variable, $R$, is
assumed to be exogenous with respect to the mediator, $M$. Moreover,
observe that these TSLS estimators of the NDE and NIE solely rely on the
TSLS estimator of $\be_{R}$.

\subsection{SPSL Estimator}\label{sec:sps}
As we have seen, the OLS and the TSLS estimators satisfy competing, yet
complementary demands. Under assumptions (OLS--1), (OLS--2), and (OLS--4), the OLS will be
asymptotically efficient but possibly biased, whereas under
assumptions (TSLS--1) to (TSLS--4), the TSLS will be
asymptotically unbiased but relatively inefficient. Thus, it is
natural to try to strike a trade-off between these two estimators, by
considering affine combinations of the form 
\begin{equation}\notag
     \bar\bbe_{\al} := \al\atsls\bbe + (1-\al)\aols\bbe,
\end{equation}
where recall that $\atsls\bbe$ and $\aols\bbe$ denote the TSLS and OLS
estimators, respectively. 
and where $\al$ needs not be comprised between 0 and 1, but may take any
real values. This family of estimators are sometimes referred to as
semi-parametric Stein-like (SPSL) estimators, for reasons which will
become clear in the sequel \citep{Judge2004}. 

In this framework, the shrinkage parameter, $\al$, is commonly
selected as the value that minimizes an empirical estimate of the MSE
of $\bar\bbe_{\al}$. 
However, in many circumstances, it may be desirable to optimize such a
trade-off with respect to a subset of the parameters of interest. This may
be achieved by pre-multiplying the vectors of estimators and estimands
with the matrix of an orthogonal projection, which will select the
particular subset of parameters that one wishes to emphasize. That is,
given a projection, $\bP$, we may consider the MSE of the vector
\begin{equation}\notag
     \bP(\bar\bbe_{\al} - \bbe) = (\bP\bar\bbe_{\al} - \bP\bbe).
\end{equation}
The shrinkage parameter, $\al$, is defined as the value that
minimizes the trace of the MSE of that projected vector, which is given by
\begin{equation}\notag 
     \tr\mse(\bP\bar\bbe_{\al})
     := \tr\E\big[\bP(\bar\bbe_{\al}-\bbe)(\bar\bbe_{\al}-\bbe)\pri\bP\pri\big].
\end{equation}
The use of a projection in this setting can be regarded as a generalization
of the original SPSL framework introduced by Judge and Mittelhammer\citep{Judge2004}. Before
turning to the minimization of that quantity, we describe a particular
decomposition of the MSE of the SPSL estimator.

Using $\bP\bar\bbe_{\al} = \al\bP\atsls\bbe + (1-\al)\bP\aols\bbe$, one
can show that the MSE of $\bP\bar\bbe_{\al}$ can be decomposed into a
weighted combination of the MSEs for the projected OLS and 
TSLS estimators. That is, for every $\al$, and every projection,
$\bP$, we obtain,
\begin{equation}\label{eq:mse decomposition}
     \mse(\bP\bar\bbe_{\al}) = 
     \al^{2}\mse(\bP\atsls\bbe) + \al(1-\al)\cse(\bP\atsls\bbe,\bP\aols\bbe)
     + (1-\al)^{2}\mse(\bP\aols\bbe),
\end{equation}
where the \textit{cross sum of squares},
$\cse(\bP\atsls\bbe,\bP\aols\bbe)$ is defined as 
$\E[\bP(\atsls\bbe-\bbe)(\aols\bbe-\bbe)\pri\bP\pri]$.
The theoretical parameter, $\al$, controlling the respective
contribution of the OLS and TSLS estimators is then defined as the
following minimizer, 
\begin{equation}\label{eq:al}
     \al := \argmin_{\al\in\R} \tr\mse(\bP\bar\bbe_{\al}).
\end{equation}
This parameter can be shown to be available in closed-form. This
follows from the fact that the MSE of $\bP\bar\bbe_{\al}$ is a convex
function of $\al$. In the following proposition, for every estimator
$\bbe^{\dagger}$, the quantity $(\tr\mse(\bbe^{\dagger}))^{1/2}$ is
referred to as the trace RMSE of $\bbe^{\dagger}$. A proof of this
proposition is provided in the appendix. 
\begin{pro}\label{pro:sps}
    For every $n$, and every $\bP$; the parameter $\al$
    from equation (\ref{eq:al}) is 
    \begin{equation}\notag
          \al = 
          \frac{\tr(\mse(\bP\atsls\bbe) - \cse(\bP\atsls\bbe,\bP\aols\bbe))}
          {\tr(\mse(\bP\atsls\bbe)-2\cse(\bP\atsls\bbe,\bP\aols\bbe)+\mse(\bP\aols\bbe))}.
    \end{equation}
    If, in addition, the random vectors, $\atsls\bbe$ and $\aols\bbe$ are
    elementwise squared-integrable, then $\al$ is unique whenever the
    trace RMSEs of $\bP\atsls\bbe$ and $\bP\aols\bbe$ are not equal. 
\end{pro}
In order to estimate the shrinkage parameter from the data, we need
to construct a consistent estimator of the bias of
$\bP\bar\bbe_{\al}$. Indeed, the MSE of that estimator can be
decomposed as follows,
\begin{equation}\notag
     \mse(\bP\bar\bbe_{\al}) = \var(\bP\bar\bbe_{\al}) +
     \bias^{2}(\bP\bar\bbe_{\al}),
\end{equation}
where $\bias^{2}(\bP\bar\bbe_{\al}):=(\E[\bP\bar\bbe_{\al}]-\bP\bbe)
(\E[\bP\bar\bbe_{\al}]-\bP\bbe)\pri$.
In general, the second term in the latter equation will not be
directly available. Nonetheless, one can show that the assumptions
that were made to guarantee the validity of the instruments used 
in section \ref{sec:tsls}, will also be sufficient to provide us with a consistent
estimator of the bias of $\bP\bar\bbe_{\al}$. Indeed, since by
proposition \ref{pro:tsls consistency}, we have seen that the TSLS
estimator converges in probability to the true parameter, $\bbe$; it
follows that this particular estimator can be used in the place of the
true parameter in order to produce a consistent estimator of the bias
of $\bP\bar\bbe_{\al}$. That is, we can define the empirical bias of
the projected SPSL estimator as follows,
\begin{equation}\notag
    \wh{\bias}(\bP\bar\bbe_{\al}) := \bP\bar\bbe_{\al}- \bP\atsls\bbe.
\end{equation}
The CSE from proposition \ref{pro:sps} can be estimated in an
analogous fashion. Therefore, the consistency of the TSLS estimator
guarantees the consistency of the SPSL estimator. 

The choice of terminology for this family of estimator can be
justified by observing that the expression for $\al$ in proposition
\ref{pro:sps} bears some similarities with the theory of Stein
estimators \citep{Efron1973}. Indeed, the empirical version of the
formula for the shrinkage parameter can be expressed as follows, 
\begin{equation}\notag
    \wh\al = \frac{\tr(\wh\var(\bP\atsls\bbe) -
      \wh\cse(\bP\atsls\bbe,\bP\aols\bbe))}{||\bP(\aols\bbe-\atsls\bbe)||^{2}},
\end{equation}
where $||\cdot||$ denotes the $L_{2}$-norm on $\R^{k+2}$, with respect
to the empirical joint distribution of the data. Using this
expression, we can then formulate the SPSL estimator as a weighted
deviation from the unbiased TSLS estimator, shrank toward
the OLS estimator, 
\begin{equation}\notag
     \bar\bbe_{\wh\al} = \atsls\bbe -
     \frac{\wh\tau}{||\bP(\atsls\bbe-\aols\bbe)||^{2}}
     (\aols\bbe - \atsls\bbe),
\end{equation}
in which $\wh\tau:=\tr(\wh\var(\bP\atsls\bbe) -
\wh\cse(\bP\atsls\bbe,\bP\aols\bbe))$, and where observe that we have made
implicit the dependence of the LHS in the latter equation on
$\bP$. Indeed, $\bar\bbe_{\wh\al}$ is solely dependent on the
projection, $\bP$, through the value of $\wh\al$, since we have
$\wh\al=\wh\tau/||\bP(\atsls\bbe-\aols\bbe)||^{2}$.

The relationship between the SPSL estimator
and the traditional Stein estimators has been studied by
previous authors. See Judge and Mittelhammer\citep{Judge2013}, for
instance. One can also observe that under the additional assumption that the random
vectors, $\atsls\bbe$ and $\aols\bbe$, are elementwise squared integrable;
it follows that we can obtain a central limit theorem for the
SPSL estimator dependent on $\bP$. This would generalize a previous result by
Judge and Mittelhammer \citep{Judge2013} for the standard SPSL estimator. 

As for the OLS and TSLS estimators, the natural causal effects of the
experimental manipulation onto the outcome, can
be estimated using the components of the SPSL estimator, $\bar\bbe_{\wh\al}$,
such that we obtain $\ov\nde := \bar\be_{R}$, and $\ov\nie :=
\wti\theta_{R}-\bar\be_{R}$; where note that, as for the TSLS natural
effects, the quantity $\wti\theta_{R}$ is still estimated using the OLS
estimator. 

For the analysis of the PROSPECT data set, since the
estimations of both the NDE and the NIE rely on this
quantity, it follows that the main
parameter of interest is $\be_{R}$. We have here arranged
the variables in this model according to $V_{i}=
(X_{i}\pri,M_{i},R_{i})\pri$. Thus, the projection matrix, $\bP$, will
be defined as a null matrix with a single non-null value in the last
element of its diagonal (that is, $P_{ij}=0$ holds every element in
$\bP$, apart from $P_{k+2,k+2}=1$); thereby estimating the shrinkage
parameter solely on the basis of the respective values taken by
$\aols\be_{R}$ and $\atsls\be_{R}$.

\section{Simulations}\label{sec:sim}
We now present a simulation study, which compares the OLS and TSLS with the
combined estimator, SPSL. We generate data from 
a confounded mediation model augmented
with an instrumental variable. The design of this simulation
experiment is partly motivated by the model fitted to the PROSPECT data set
analyzed in the sequel. Note, however, that in our simulations, the
mediator is assumed to be continuous, whereas that same variable is
dichotomous in the PROSPECT data. The effect of treatment on the
endogenous mediator is allowed to vary according to the values taken
by the baseline variables. Apart from this source of variation, the
effects are assumed to be homogeneous in these simulations. 
\begin{figure}[htbp]
  \centering
  \textbf{(A) Natural Direct Effect}\\
  \includegraphics[width=13cm]{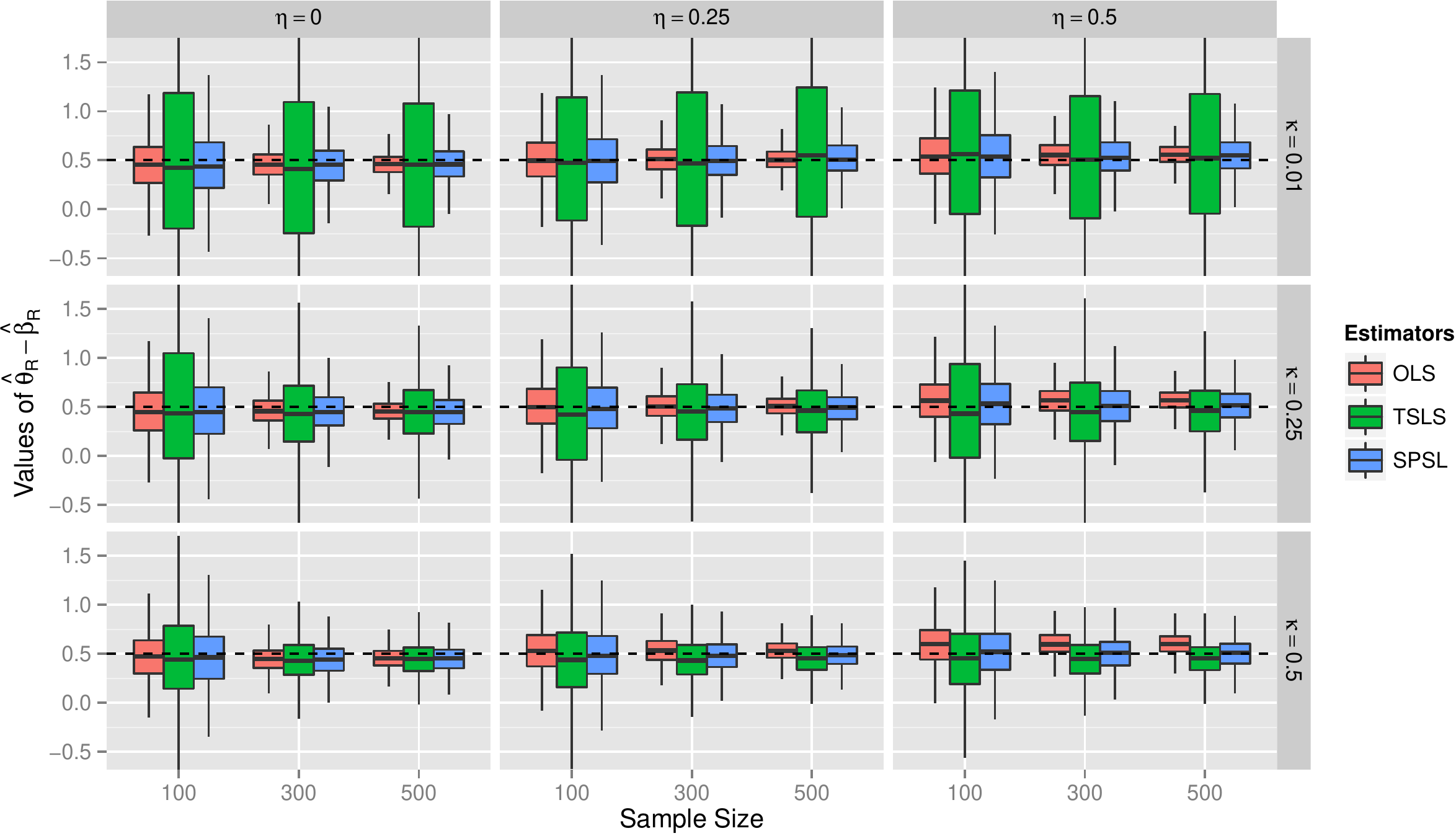}
  \vspace{.25cm}\\
  \textbf{(B) Natural Indirect Effect}\\
  \includegraphics[width=13cm]{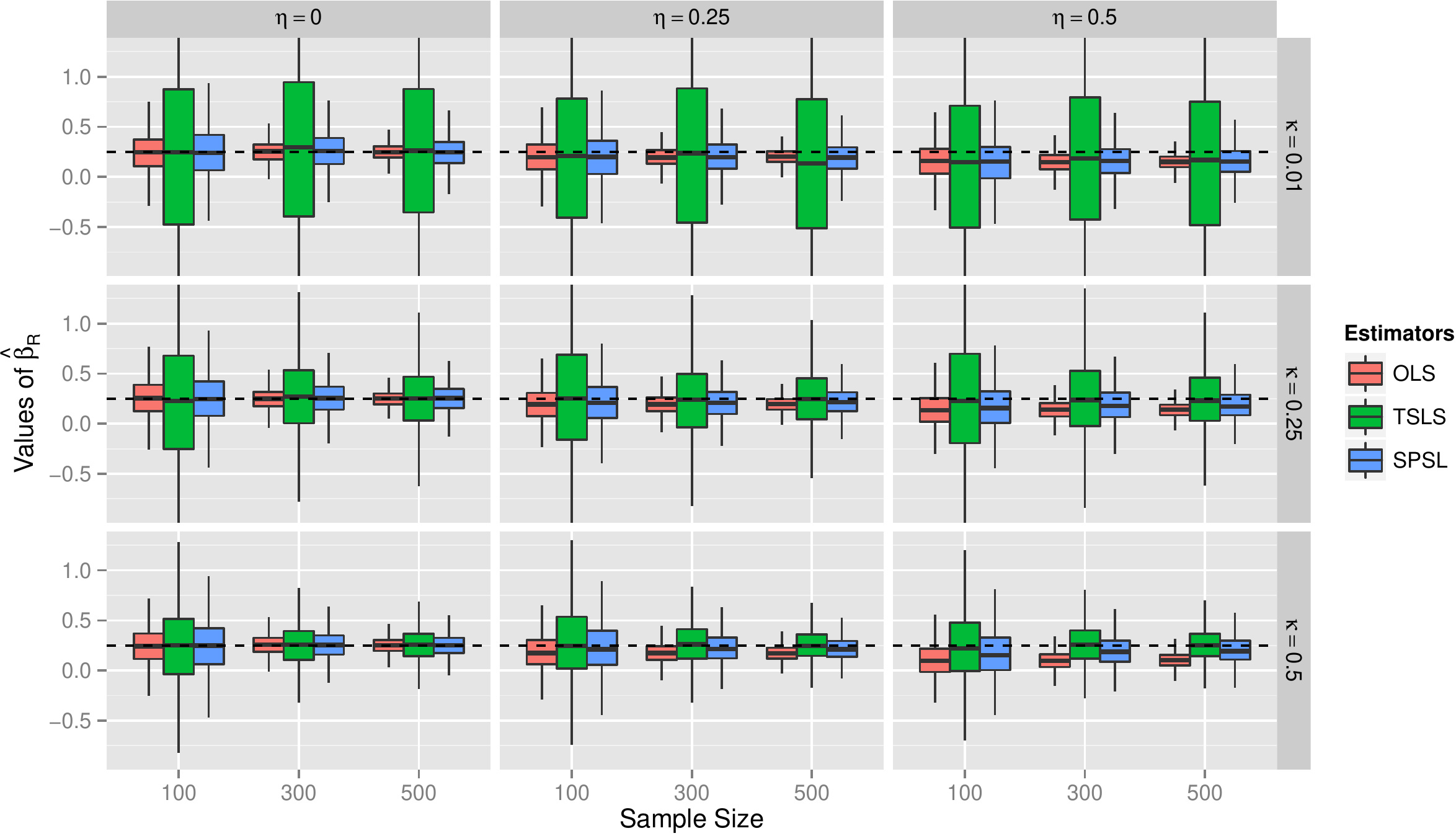}
   \caption{Monte Carlo distributions of estimators' values
     of the three estimators of interest under the simulation scenarios described
     in Figure \ref{fig:instrumented model}, for the NDE, $\be_{R}$,
     and NIE, $\theta_{R}-\be_{R}$, in
     panels (A) and (B), respectively. The simulations are reported
     for different degrees of confounding, and varying levels of
     instrument's strength, measured by $\eta$ and $\kappa$,
     respectively. These results are based on $10^5$ iterations in
     each condition. The dashed lines indicate the values of the true
     NDE and NIE, in panels (A) and (B), respectively.
     \label{fig:beta}}
\end{figure}

\subsection{Mediation Model}\label{sec:sim model}
Our objective in constructing our simulation model is
twofold. Firstly, we wish to be able to control the degree of
endogeneity of the mediator, as well as the strength of the instrument;
such that both factors can be varied independently of each
other. Secondly, we will also require the variances of the response,
$Y_{i}$'s, and of the intermediate variable, $M_{i}$'s, to be equal to
1, to be able to interpret the size of the effect on a standardized
scale. 

As represented in Figure \ref{fig:instrumented model}, we formulate
the following structural model for the clinical outcome,
\begin{equation}\notag
   Y_{i} = \be_{X}X_{i} + \be_{R}R_{i} + \be_{M}M_{i} +
      \be_{U}U_{i} + \ep_{i};
\end{equation}
for every $i=1,\ldots,n$. (Note that, contrary to the model in
Equation \eqref{eq:model}, the $\ep_{i}$'s in this simulation model
are uncorrelated with the $U_{i}$'s.)
As previously mentioned, in order facilitate
interpretability, we will fix the variance of the response variable to
be equal to 1 for all scenarios. The variance of the intermediate
variable, $M_{i}$'s, will also be constrained to be unity. Both
of these objectives will be achieved by controlling the variances of the error
terms, $\ep_{i}$'s in the above model; and $\de_{i}$'s in the
following model for the intermediate variable, 
\begin{equation}\notag
     M_{i} = \ga_{X}X_{i} + \ga_{R}R_{i} + \ga_{RX}
              R_{i}X_{i} + \ga_{U}U + \de_{i}.
\end{equation}
(Note again that the $\de_{i}$'s in the above simulation model for the
mediator are uncorrelated with the $U_{i}$'s.)
The variance of the $\de_{i}$'s is defined as a function of the
parameters in the equation for the $M_{i}$'s, such that 
$\sig_{\de}^{2}(\ga_{X},\ga_{R},\ga_{RX},\ga_{U}):=\var(\de)$. This
function will be defined in the sequel. For convenience, we will
simulate a single baseline covariate, denoted by $X_{i}$. This
baseline covariate is given the following distribution,
$X_{i}\stack{\iid}{\sim}N(0,2)$; where the variance was arbitrarily
fixed to two, in order to simplify some of our computations. In addition,
the experimental factor is drawn from a Bernoulli distribution, taking
the form, $R_{i}\stack{\iid}{\sim}\bern(1/2)$. Finally, the unmeasured
confounder is also generated from a unit normal distribution, such
that $U_{i}\stack{\iid}{\sim}N(0,1)$.

In this model, the $X_{i}$'s are assumed to be independent of other
observed baseline variables, such that $X_{i}\perp
R_{i}$; and the confounders, denoted by $U_{i}$'s, are assumed to
solely affect the relationship between the outcome and the mediator,
such that we also have $U_{i}\perp X_{i},R_{i},R_{i}X_{i}$. These
assumptions, combined with our constraints on the variances of the
$Y_{i}$'s and the $M_{i}$'s, can be used to compute a range of
possible values for the parameters of interest. A description of the
specific computations involved in this derivation has been relegated
to an appendix. (See Appendix B, for the details of the computation of
the variance of the error terms, $\sig^{2}_{\ep}$ and $\sig^{2}_{\de}$.)
Throughout these simulations, the parameters controlling the effect of
the $X_{i}$'s and $R_{i}$'s have been set to $\ga_{X}:=1/4$, and
$\ga_{R}:=1/\sqrt{2}$, respectively. These choices of parameters
correspond to small to moderate effect sizes. 
For convenience, we have further set the coefficients of
the structural model for the $Y_{i}$'s to take the same value, 
$\be_{X}=\be_{R}=\be_{M}=\be_{U}=1/4$. It then follows that in order to guarantee
$\sig^{2}_{\ep}>0$, we need to choose $\ga_{U}$, as satisfying
$\ga_{U}\leq1/2$, as well as, $\ga_{X}+\ga_{RX}\leq1/2$. 

The two main factors that are manipulated in
this simulation study are the degree of confounding of the mediator, and the
strength of the instrument. These simulation factors are respectively quantified
using the correlation of the intermediate variable, $M_{i}$'s, with
the confounders, $U_{i}$'s; and with the instruments,
$R_{i}X_{i}$'s. Owing to our choice of normalization, these
two correlations can be expressed as follows,
\begin{equation}\notag
     \Cor(M_{i},U_{i}) = \eta, 
     \qquad\te{and}\qquad
     \Cor(M_{i},R_{i}X_{i}) = \kap; 
\end{equation}
where it can be verified that $\eta=\ga_{U}$, and $\kap=\ga_{X}+\ga_{RX}$.
Under the additional constraint that both $\sig^{2}_{\ep}$ and
$\sig^{2}_{\de}$ are positive, it follows
that we can select $\eta$ to take values in the set
$\{0.0,0.25,0.50\}$, which represent different choices for the
degree of confounding (none, moderate, and
strong on a correlation scale); 
and $\kap$ to take values in the set $\{0.01,0.25,0.50\}$, which
represent different choices for the strength of the
instrument (weak, moderate, and strong also on a correlation scale).
Observe that the correlation between the mediator and its
instrument, $\ka$, must be non-zero; in order to ensure that the TSLS
estimator is well-identified in all scenarios. 
\begin{figure}[htbp]
  \centering
  \textbf{(A) Natural Direct Effect}\\
  \includegraphics[width=13cm]{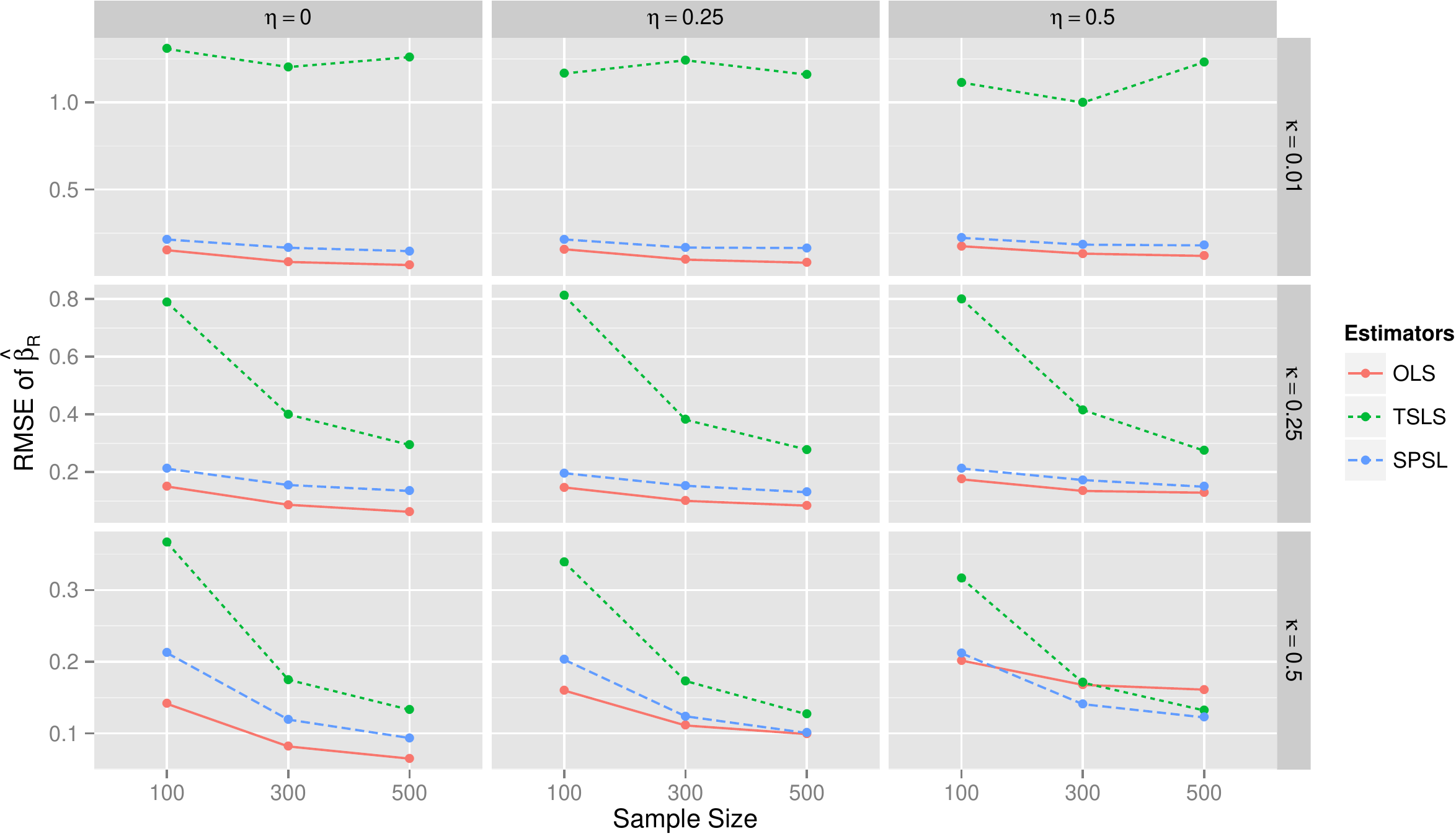}
  \vspace{.25cm}\\
  \textbf{(B) Natural Indirect Effect}\\
  \includegraphics[width=13cm]{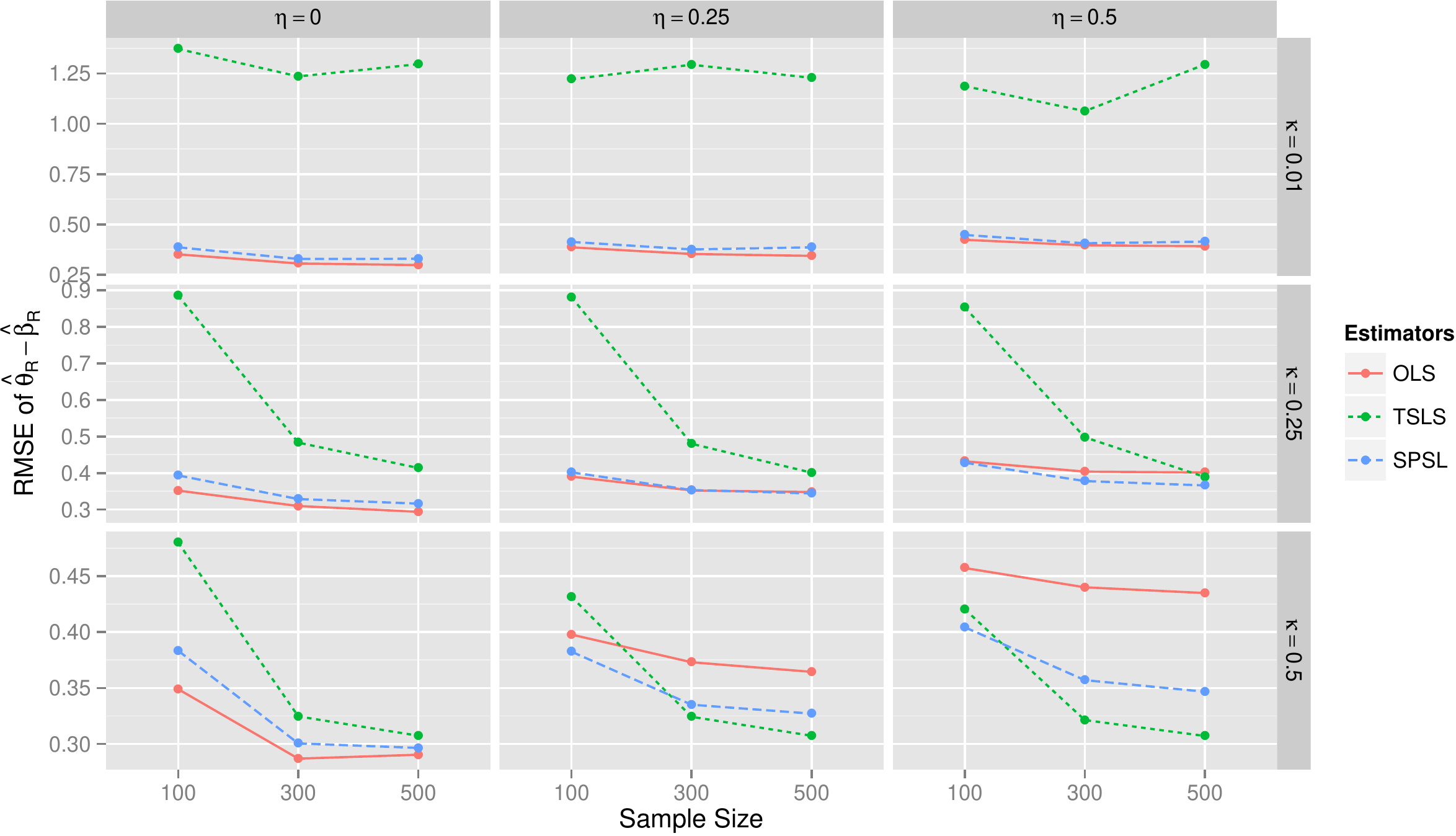}
   \caption{Monte Carlo estimates of the root mean squared errors (RMSEs)
     of the three estimators of interest under the simulation scenarios described
     in Figure \ref{fig:instrumented model}, for the NDE, $\be_{R}$,
     and NIE, $\theta_{R}-\be_{R}$, in
     panels (A) and (B), respectively. The simulations are reported
     for different degrees of confounding, and varying levels of
     instrument's strength, measured by $\eta$ and $\kappa$,
     respectively. These results are based on $10^5$ iterations in
     each condition. 
     \label{fig:mse}}
\end{figure}

\subsection{Evaluation of the Estimators}\label{sec:sim evaluation}
We generated $10^5$ Monte Carlo samples from the aforementioned model,
under combinations of the three values taken by $\eta$, the three values taken by
$\kap$; and the three different sample sizes typical of mental health trials,
$n\in\{100,300,500\}$. Altogether, this produced a total of $270,000$ distinct
synthetic data sets.

The OLS, TSLS and SPSL estimators of the NIE and NDE were computed as follows.
Firstly, for each data set, we computed the OLS estimator, $\wti\theta_{R}$ of
the total effect of $R$ on the outcome $Y$. This corresponds to estimating
the non-mediated model presented in Equation \eqref{eq:total},
and illustrated in Figure \ref{fig:total}. Observe
that the OLS estimator, $\wti\theta_{R}$, is identical for all methods
of estimation. Indeed, the estimation of the total effect in this
model is assumed to be unbiased, since
subjects have been randomly allocated to the levels of the
experimental factor, $R$.

Secondly, we fitted the instrumented mediation model, corresponding to the
diagram in Figure \ref{fig:instrumented model}, for the three different
estimation procedures. This produced the OLS, TSLS and SPSL estimators
for $\be_{R}$, which corresponds to the estimator of the NDE.
The NIE estimator could then be obtained by subtracting that estimate from $\wti\theta_{R}$. 
The Monte Carlo distributions of these quantities for the three estimators
under the scenarios considered are plotted in Figure \ref{fig:beta}.
The performances of these estimators were also compared by computing
the empirical root MSE over the $10^{5}$ Monte Carlo samples generated
in each scenario. These RMSEs are reported in Figure \ref{fig:mse}.

\subsection{Simulation Results}\label{sec:sim result}
Consider the distribution of the values taken by the three
estimators of interest in Figure \ref{fig:beta}. These are reported
for the two causal estimands under scrutiny: NIE,
$\wti\theta_{R}-\be_{R}^{\dag}$, and NDE, $\be_{R}^{\dag}$, in which $\be_{R}^{\dag}$
may represent either the OLS, TSLS or SPSL estimators. 
As expected, for both the NIE and NDE, the OLS was more likely to be biased when
$\eta$ was large, and the TSLS was more likely to exhibit a high
variance when $\kap$ was small. Increases in sample size tended to
result in better precision for all estimators. This trend was
particularly noticeable for the TSLS estimator. 

The overall performances of these estimators were also compared using their respective
RMSEs. These have been reported in Figure \ref{fig:mse}. The patterns
exhibited by the NDE and NIE were very similar. For weak instrumental
variables, i.e.~$\kap=0.01$, the RMSE of the TSLS estimator was high in
comparison to the ones of the OLS and SPSL estimators, due to the
large variance of the TSLS. In these scenarios, the Stein-like estimator's
RMSE was almost identical to the one of the OLS. By contrast, when the
instrumental variables were strongly predictive of the endogenous
variable, i.e.~$\kap=0.5$, the RMSEs of the different estimators
varied with the amount of bias. This trend is particularly
noticeable in the last row of Figure \ref{fig:mse}(A). Fixing the
correlation between the mediator, $M$, and the confounder, $U$, to be $\eta=0.5$;
one can observe that for small values of $\kappa$, the RMSE of the OLS is
optimal, whereas for large values of $\kappa$, the RMSE of the TSLS is
optimal; while the SPSL strikes a trade-off between these two
counterparts irrespective of the values taken by $\kappa$. 

In practice, we can usually evaluate the strength of a set
of instruments, by computing the $F$-test of the equation for the
$M_{i}$'s. In our simulation, this corresponds to having some
knowledge of $\kappa$. However, it is generally not possible to obtain
any information about the degree of confounding, $\eta$. These
simulations have therefore demonstrated that the SPSL outperforms its
counterparts in a \emph{global} sense --that is, when we `average' the
performances of these estimators over different values of
$\eta$. Intuitively, this approach bears some similarities with the
Bayesian framework for model averaging, in which the degree of
unmeasured confounding, $\eta$, is treated as a source of
uncertainty. 

\section{PROSPECT Study}\label{sec:prospect}
We here re-analyze a randomized controlled trial known as PROSPECT
\citep{Bruce2004}. This study tested the impact of a primary
care intervention on reducing major risk factors for suicide in
late life. Patients were recruited from 20 different primary
care practices on the East coast of the United-States, over a 16-month
period. The intervention consisted in two major components\citep{Bruce2004}. Firstly,
the physicians followed a clinical algorithm specifically designed for
treating geriatric depression. Secondly, the treatment was managed and
adjusted by depression care managers. This primary care intervention
was compared to a treatment as usual (TAU) condition. 
\begin{figure}[t]
  \centering
  \includegraphics[width=7cm]{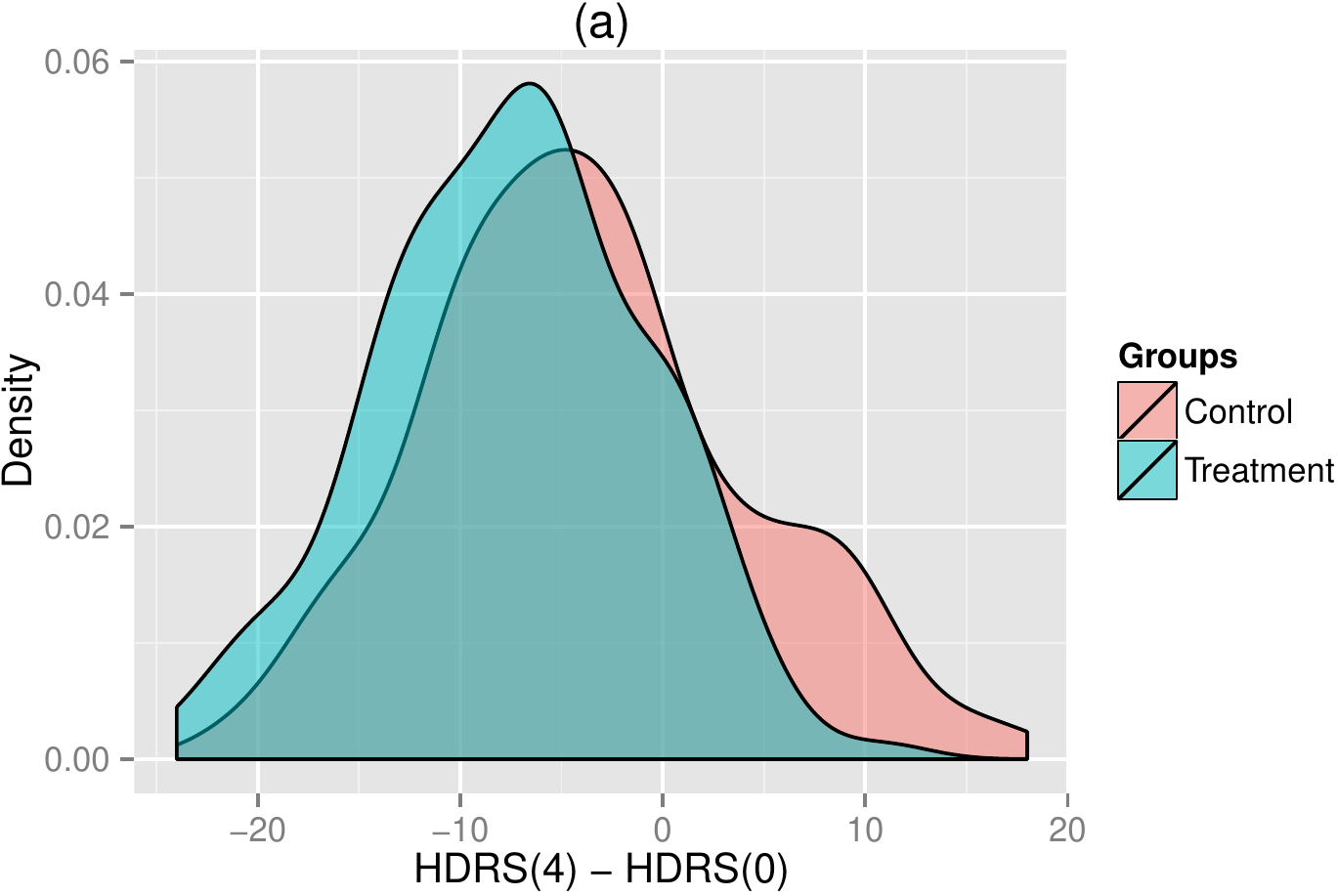}
  \includegraphics[width=7cm]{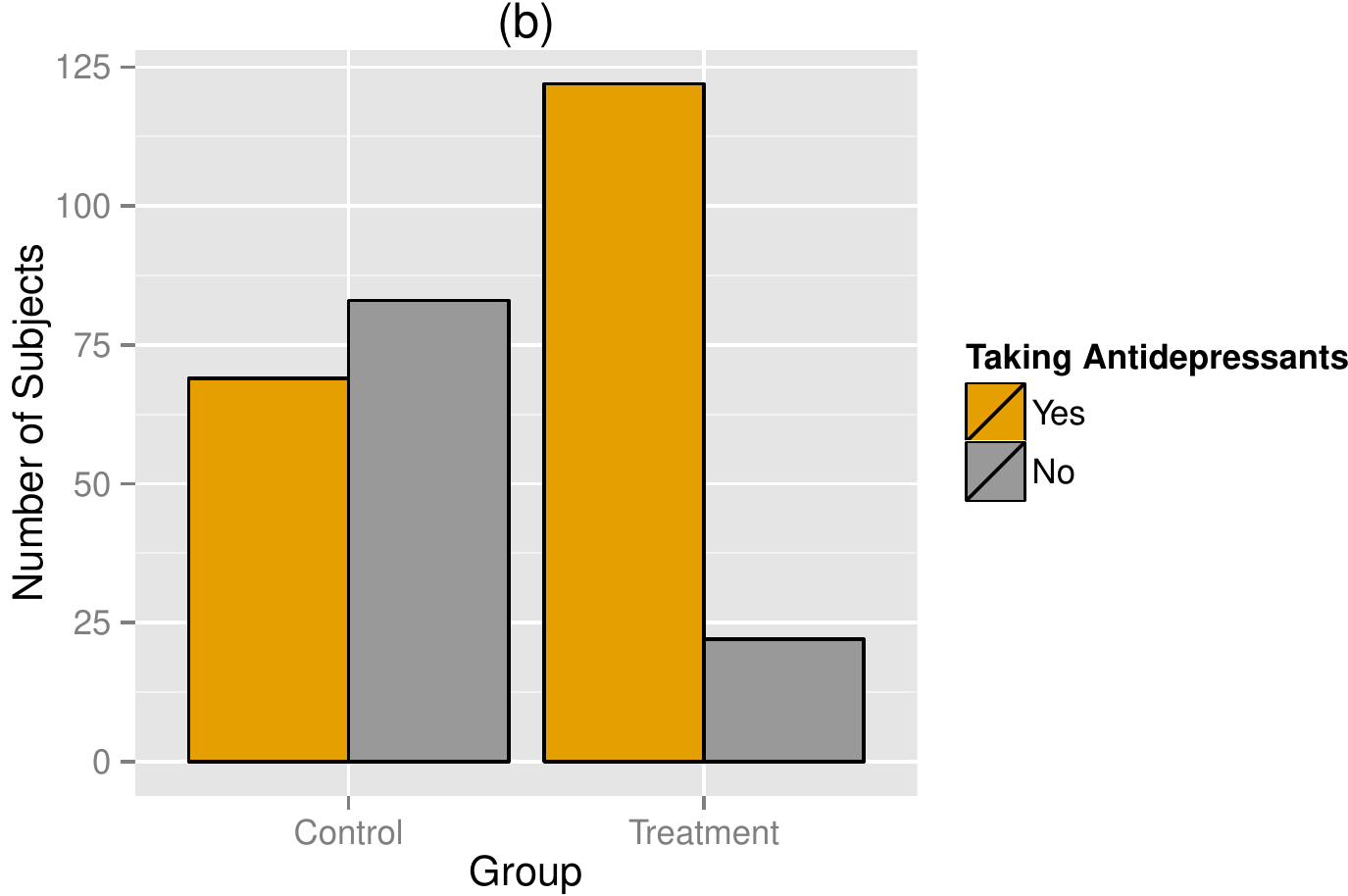}
  \caption{Descriptive statistics for the PROSPECT data set. In panel
    (a), we have provided histograms of the difference,
    $\te{HDRS}(4)-\te{HDRS}(0)$, for both
    the control and treatment groups, where HDRS(0) and HDRS(4) denote
    the Hamilton Depression
    Rating Scale (HDRS) at baseline and after a four-month
    follow-up. In panel (b), the barplots represent the distribution of
    patients according to whether or not they have been taking
    antidepressants, which here corresponds to the intermediate
    variable, $M$, reported by treatment groups, $R$.}
    \label{fig:prospect}
\end{figure}

\subsection{Mediation Model}\label{sec:prospect mediation}
The main question of interest here is to investigate whether the
intent-to-treat effect of the intervention on the 4-month Hamilton
Depression Rating Scale (HDRS) was due to a direct effect of
treatment allocation, after excluding the indirect effect
mediated through taking antidepressant medication. Thus, the intermediate variable in this
study should be regarded as a nuisance mediator, since we are
primarily interested in the direct effect. The PROSPECT mediation study is
therefore unusual, in the sense that the main effect of interest
in the present analysis, is not the indirect effect or NIE as in most
mediation studies. 

The subjects' scores on the HDRS
after a four-month follow-up is the main outcome under scrutiny.
The instrumental variables for the mediator were defined as the set
of interaction terms between the randomized intervention and the baseline
covariates. This particular choice of instruments has been proposed
previously\citep{TenHave2007}, and we are here following this choice for
comparability; see also Small \citep{Small2012}
for a discussion of the use of interaction terms as instruments in the
context of causal mediation. The
instruments were found to be good predictors of the endogenous
mediator; and explained about 50\% of the variance in that
variable. Fitting a linear regression with taking prescribed
antidepressant medication as the 
response, and the instruments as predictors, resulted in a highly
significant $F$-statistic $(F=9.10,\df_{1}=6,\df_{2}=282,p<0.001)$,
thereby justifying our choice of instruments for this study. See also
a similar analysis of the strength of the instruments in this study in
Emsley et al.~(2010)\citep{Emsley2010}.

The baseline variables included HDRS scores at baseline, denoted
HDRS(0), a binary variable denoting suicide ideation, 
past medication use (i.e.~whether or not patients had been using past
medication for dementia and other conditions but excluding
psychotropic treatment for depression), and antidepressant use
(i.e.~specifically whether or not patients had been using antidepressant
medication in the past). Moreover, the model also included two dummy variables,
which controlled for the three different collection sites that were
used in the study. Descriptive statistics for the main variables
of interest in this study, have been reported in Figure \ref{fig:prospect}.
\begin{table*}[t]
\centering
\caption{Re-analysis of the PROSPECT data set$^{a}$\tnote{a}, in which the
  outcome variable is the Hamilton Depression Rating Scale at
  four-month, HDRS(4), the main intervention is the primary care
  intervention of interest, whereas the mediator is taking
  antidepressant medication. Three estimators of interest are here compared. These include the
  Ordinary Least Squares (OLS), Two-Stage Least Squares (TSLS), and
  the Semi-Parametric Stein-Like (SPSL) estimators. 
  Bootstrapped standard errors$^{b}$\tnote{b} for all estimators are denoted in
  parentheses.\label{tab:prospect}}
  \setlength{\tabcolsep}{0.1cm}
  \renewcommand{\arraystretch}{1.15}
\begin{3table}
\begin{tabular}{@{}>{\centering}m{10pt}>{\raggedright}m{150pt}>{\centering}m{1pt}|
                >{\raggedleft}m{60pt}>{\raggedleft}m{60pt}>{\raggedleft}m{60pt}m{1pt}@{}}
   \toprule
   \mcol{2}{c}{\textit{Variables in Model for $Y$}} &&
   \mcol{1}{c}{OLS} & \mcol{1}{c}{TSLS} & \mcol{1}{c}{SPSL} & \\
   \midrule
\mcol{2}{l}{\textit{Randomization \& Mediator:}}&& & & & \tabularnewline
&$R:$ Primary Care Intervention&&$-$2.66 (0.96) & $-$2.38
(1.37)& $-$2.40 (1.07)\tabularnewline
&$M:$ Antidepressant Medication&&$-$1.24 (1.09) &$-$1.95 (2.56)& $-$1.90 (1.48)\tabularnewline
\mcol{2}{l}{\textit{Baseline Covariates:}}&& & & & \tabularnewline
&$X_{1}:$ HDRS(0)&&0.62 (0.07) & 0.62 (0.07) & 0.62 (0.07)\tabularnewline
&$X_{2}:$ Suicide Ideation &&1.25 (0.96)&1.25 (0.96)& 1.25 (0.94)\tabularnewline
&$X_{3}:$ Past Medication Use&&1.48 (1.07)&1.59 (1.10)&1.58 (1.05)\tabularnewline
&$X_{4}:$ Antidepressant Use&&$-$0.14 (0.40)& $-$0.07 (0.44)& $-$0.07 (0.40)\tabularnewline
&$X_{5}:$ Second Collection Site&&$-$0.46 (0.99) &$-$0.50 (0.97)& $-$0.49 (1.01)\tabularnewline
&$X_{6}:$ Third Collection Site&&$-$2.13 (1.05)& $-$2.05 (1.16)&$-$2.05 (1.04)\tabularnewline
\midrule
\mcol{2}{l}{\textit{Causal Effects:}}&& & & & \tabularnewline
&NDE: Natural Direct Effect &&$-$2.66 (0.96) & $-$2.38 (1.37)& $-$2.40 (1.07)\tabularnewline
&NIE: Natural Indirect Effect && $-$0.48 (1.26) & $-$0.76 (1.60) & $-$0.74 (1.35)\tabularnewline
\mcol{2}{l}{\textit{Shrinkage Parameter:}}&& & & & \tabularnewline
&SPSL's $\wh{\al}$\tnote{c}&&$-$$-$&$-$$-$&0.924&\tabularnewline
   \bottomrule
\end{tabular}
\begin{tablenotes}
    \item[a] Complete cases, for whom all measures were available,
      $n=296$.
    \item[b] The SEs for all estimators are based on $1,\!000$
      bootstrap iterations. 
    \item[c] The estimated shrinkage used in the computation of
      the SPSL estimator, where the optimal shrinkage is estimated
      using a projection matrix, $\bP$, which is specified to be a
      unit matrix with a single one in its diagonal, corresponding to
      the offer of treatment variable, $\be_{R}$.
\end{tablenotes}
\end{3table}
\end{table*}
\subsection{Results of Re-analysis}\label{sec:prospect results}
The results of this re-analysis are reported in Table
\ref{tab:prospect}. The natural direct and natural indirect effects have been
computed for the three estimators of interest. The OLS estimates and
their standard errors were found to be approximately identical to the
ones reported by Ten Have and colleagues in a previous analysis of the
same data set \citep{TenHave2007}.

In this paper, we have introduced a generalization of the
SPSL estimator, which includes the use of a projection matrix, $\bP$.
Such a matrix permits to restrict the computation of the shrinkage
parameter, $\al$, to a specific subset of variables. In the case of
PROSPECT, the intermediate variable in this study is treated as a
nuisance mediator, in the sense that the main focus of the analysis lies
in estimating the NDE. Therefore, one can
select a projection matrix, $\bP$, that emphasizes the
estimation of the NDE. This can be done by
specifying $\bP$ to be a unit matrix with a single one in its
diagonal, corresponding to the offer of treatment,
$\be_{R}$. The results for the projected SPSL
estimator are reported in the third column of Table
\ref{tab:prospect}. The values of the SPSL and associated
shrinkage estimator did not markedly differ, when using an identity
matrix (results not shown); thereby indicating that the amount of shrinkage exerted by the
SPSL estimator was mostly determined by the OLS and TSLS estimates of
$\be_{R}$, the main estimand of interest. 

As expected, the values
taken by the NDE and NIE under the SPSL framework were located between
the ones of the OLS and TSLS estimators. Similarly, the standard errors
(SEs) of the SPSL estimator was also found to strike a trade-off between the SE
of the OLS estimator and the SE of the TSLS estimator, for both the NDE and NIE. The shrinkage
parameter of the SPSL estimator was found to be close to unity,
$\wh{\al}=0.924$. Thus, the bias and variance of the TSLS estimator was
favored over the corresponding properties of the OLS estimator. This
suggested that, thanks to the strength of the instruments used in this
study, we obtained a decrease in bias, which solely coincided with a
moderate increase in the variance of the TSLS.

From a practical perspective, our re-analysis indicates that a naive
use of the OLS tends to lead to an over-estimation of the direct effect of
the intervention; which therefore downplays the importance of the use of
antidepressant medication. 

\section{Conclusions}\label{sec:conclusion}
We have here demonstrated the usefulness of the SPSL estimator in the
context of causal mediation analysis. This implementation has also generalized some of
the previous uses of this family of estimators, by restricting the
estimation of the shrinkage parameter to a subset of the parameters
of interest. Although the strength of a set of instruments can usually
be estimated from
the data; the degree of unmeasured confounding is, by definition,
unknowable. In such a context, the SPSL framework produces an
estimation, which on average, will outperform both the OLS and the
TSLS, in terms of MSE. 

Furthermore, the SPSL estimator possesses desirable asymptotic
properties. Under standard assumptions on the properties of the
instruments, the SPSL estimator is indeed asymptotically unbiased.
It also has the advantage of being directly estimable from
the data. Moreover, the shrinkage parameter used in combining the
OLS and TSLS estimators may be of special interest. This parameter
can be interpreted as a gauge that measures the usefulness of the
instruments, in terms of gains in MSE. That is, a very low value for
$\wh\al$ indicates that the OLS is preferable over the TSLS, and
therefore that the corresponding instruments mostly contribute to
increasing the variance of the combined estimator, without substantial
gains in terms of unbiasedness.

Throughout this paper, we have assumed that the IVs of interest were
valid instruments. In particular, we have required that each IV only
affects the clinical outcome through the intermediate variable.
Moreover, these
assumptions have been tailored to the case in which the instruments are
constructed by interacting some of the baseline variables with
treatment offer, following the work of Small\cite{Small2012}.
Such assumptions are particularly important in our context, since the
asymptotic unbiasedness of the SPSL estimator solely holds, when the
TSLS estimator is also guaranteed to be asymptotically unbiased. Note,
however, that the SPSL framework is more widely applicable, and could
be used with instruments that are not necessarily
composed of interaction terms. 

It is of special interest to consider the behavior of the SPSL
estimator, when some of our assumptions fail to be satisfied. Let us
first focus on some of the aforementioned OLS and TSLS assumptions.
We can evaluate how the violation of these assumptions would impact on
the behavior of the SPSL estimator. In the first instance, consider
condition (OLS--4), which requires that $\E[X_{i}\pri X_{i}]$ should
be full-rank; or equivalently that the OLS estimator is identified.
If such an assumption were to fail, then the condition number of the matrix,
$\E[ X_{i}\pri X_{i}]$, would be very high, and consequently the
determinant of its inverse would be very large. As a result, this
would produce a large OLS variance, possibly tending towards
infinity. Therefore, everything else being equal, the failure of
(OLS--4) would be likely to put the OLS at a disadvantage, in
comparison to the TSLS, in the construction of SPSL.

Similarly, when considering the TSLS estimator, one can also predict
the consequences of the violation of certain assumptions.
Consider (TSLS--3), for instance. This assumption requires that the
matrices, $\E[Z_{i}\pri Z_{i}]$ and $\E[Z_{i}\pri X_{i}]$, are both
full-rank; which guarantees that the TSLS estimator is identified. If
this condition were to fail, we would obtain a very large variance for
the TSLS estimator. As a result, a failure of assumption (TSLS--3)
would then lead to the SPSL estimator being shrank toward the value of
the OLS estimator. Moreover, the TSLS would also suffer if
condition (TSLS--4) were to fail. This assumption requires that the
instruments, $Z_{i}$'s, are relevant, in the sense that they should be
correlated with the intermediate variables, $M_{i}$'s. If this
assumption were to be violated, the instruments would solely
contribute to TSLS by increasing its variance; thereby making it more
likely for the combined estimator to favor the OLS estimator over its
TSLS counterpart. 

Observe that all of the assumptions that we have made in this paper
have also been posited by Small (2012)\cite{Small2012}, in his
investigations of the properties of IVs which are defined as
interaction terms between baseline variables and the experimental variable. This choice of IVs
corresponds to the instruments that we have used in the PROSPECT data
set. In this setting, the IV assumptions could be subjected to
a sensitivity analysis, as demonstrated by Small \cite{Small2012}; and
we refer the reader to this paper for further details on the type of
sensitivity analysis that can be conducted, when using interaction
terms as instruments. However, further research will be needed to 
generalize these sensitivity analyses to the case of the SPSL
estimator. 

The SPSL could straightforwardly be applied to other causal estimands. 
It has been used to optimize the estimation of Local Average Treatment
Effects (LATEs) \cite{Ginestet2017a}, and could be implemented in
other settings. Moreover, such methods could also accommodate other
families of estimators, such as the jackknife IV estimator
(JIVE)\cite{Angrist1995}, for instance.
Further research may also concentrate on extending the applicability of the
present methods to data sets with binary outcomes. Such
extensions will need to rely on the use of IVs in generalized linear
models. Several authors have proposed methodological
frameworks for allowing the use of instruments in this
context\cite{Clarke2010,Clarke2012}; and the Stein-like estimators
could be adapted to generalized linear models using the approaches
advocated by these authors.

\section{Appendix A: Proofs of Propositions}
\begin{proof}[Proof of Proposition \ref{pro:uncorrelated}]
    The error term, $\ep_{i}$, has been defined in equation
    (\ref{eq:ep}), and can be seen to be the sum of three distinct
    random variables centered at zero. Thus, we solely need to consider
    whether or not the interaction terms, $R_{i}X_{i}$'s, are
    uncorrelated with each of the summands composing the $\ep_{i}$'s
    in equation \eqref{eq:ep}.
    Indeed, whenever a random variable is pairwise uncorrelated with a set of
    random variables, it is also uncorrelated with the sum of these
    variables. We will therefore consider the three summands of $\ep_{i}$ in
    turn. These are $A_{i1}:=\big(Y_{i}(0,0) - \E[Y_{i}(0,0)|X_{i}]\big)$,
    $A_{i2}:=\big(\be_{M,i}-\be_{M}\big)M_{i}$ and 
    $A_{i3}:=\big(\be_{R,i}-\be_{R}\big)R_{i}$.

    Furthermore, observe that the covariance of $\ep_{i}$ with the
    interaction term is a vector of order $(k\times1)$. Hence, the
    proposition is proved, if we are able to show that for each of the
    $j\tth$ component of $R_{i}X_{i}$, we have
    \begin{equation}\notag
        \cov(R_{i}X_{ij},\ep_{i}) 
        = \cov\!\bigg(R_{i}X_{ij},\sum_{l=1}^{3}A_{il}\bigg) = 0,
    \end{equation}
    where $j=1,\ldots,k$. Thus, we fix an arbitrary component, say
    $R_{i}X_{ij}$, in the sequel; and consider its covariance with
    each of the three summands of $\ep_{i}$.  

    Firstly, for $A_{i1}$, observe that the covariance 
    $\cov(R_{i}X_{ij},A_{i1})$ can be expressed as the difference, 
    $\E[R_{i}X_{ij}A_{i1}]-\E[R_{i}X_{ij}]\cdot\E[A_{i1}]$.
    Using the tower rule, we have
    \begin{equation}\notag
       \E[A_{i1}] = \E\big[Y_{i}^{0,0} - \E[Y_{i}^{0,0}|X_{i}]\big]
                 = \E\big[Y_{i}^{0,0}\big] - \E\big[\E[Y_{i}^{0,0}|X_{i}]\big]=0,
    \end{equation}
    where for convenience, we have defined the shorthand, $Y_{i}^{0,0}:=Y_{i}(0,0)$.
    It then suffices to consider the quadratic term,
    $\E[R_{i}X_{ij}A_{i1}]$, which simplifies as follows, 
    \begin{equation}\notag
      \begin{aligned}
        \E\big[R_{i}X_{ij}(Y_{i}^{0,0}-\E[Y_{i}^{0,0}|X_{i}])\big]
        =\E\big[R_{i}\big]\cdot\E\big[X_{ij}(Y_{i}^{0,0}-\E[Y_{i}^{0,0}|X_{i}])\big],
      \end{aligned}   
    \end{equation}
    using the fact that the $R_{i}$'s have been randomized. 
    Through another application of the tower rule, the second term on
    the RHS of the latter equation gives
    \begin{equation}\notag
         \E\big[\E[X_{ij}(Y_{i}^{0,0}-\E[Y_{i}^{0,0}|X_{i}])|X_{i}]\big]
         = \E\big[X_{ij}(\E[Y_{i}^{0,0}|X_{i}]-\E[Y_{i}^{0,0}|X_{i}])\big]
         = 0.
    \end{equation}

    Secondly, considering the covariance of $R_{i}X_{ij}$ with the
    second summand of the error term, $A_{i2}$; we can again apply the tower
    rule in order to obtain 
    \begin{equation}\notag
      \begin{aligned}
        \E\big[(\be_{M,i}-\be_{M})M_{i}\big] 
        &= \E\big[\E[(\be_{M,i}-\be_{M})M_{i}|R_{i},X_{i}]\big] \\
        &= \E\big[\E[(\be_{M,i}-\be_{M})|R_{i},X_{i}]\cdot\E[M_{i}|R_{i},X_{i}]\big], \\
      \end{aligned}
    \end{equation}
    where the second equality is a consequence of
    assumption (A3). Moreover, the first term inside the expectation
    on the RHS of the latter equation becomes, 
    \begin{equation}\notag
       \E[(\be_{M,i}-\be_{M})|R_{i},X_{i}] 
       = \E[(\be_{M,i}-\be_{M})|X_{i}]
       = \big(\E[\be_{M,i}|X_{i}]-\be_{M}\big) = 0,
    \end{equation}
    using in turn, the fact that the $R_{i}$'s are randomized, and
    assumption (A2). Thus, the covariance,
    $\cov(R_{i}X_{ij},A_{i2})$, reduces to the quadratic term 
    $\E[R_{i}X_{ij}A_{i2}]$. However, using the tower rule, this
    quantity can be expressed as 
    \begin{equation}\notag
      \begin{aligned}
        \E[R_{i}X_{ij}(\be_{M,i}-\be_{M})M_{i}]
        &= \E\big[\E[R_{i}X_{ij}(\be_{M,i}-\be_{M})M_{i}|R_{i},X_{i}]\big]\\
        &= \E\big[R_{i}X_{ij}\cdot\E[(\be_{M,i}-\be_{M})M_{i}|R_{i},X_{i}]\big]\\
        &= \E\big[R_{i}X_{ij}\cdot\E[(\be_{M,i}-\be_{M})|R_{i},X_{i}]\cdot\E[M_{i}|R_{i},X_{i}]\big]\\
        &=0,
      \end{aligned}
    \end{equation}
    where the third equality follows from assumption (A3).

    Thirdly, the covariance, $\cov(R_{i}X_{ij},A_{i3})$ with
    $A_{i3}=(\be_{R,i}-\be_{R})R_{i}$, can similarly be simplified
    by applying the fact that the $R_{i}$'s are randomized, such that 
    \begin{equation}\notag
       \E[(\be_{R,i}-\be_{R})R_{i}]  
       = \E[(\be_{R,i}-\be_{R})]\cdot\E[R_{i}] 
       = \big(\E[\be_{R,i}]-\be_{R}]\big)\cdot\E[R_{i}] 
       = 0,
    \end{equation}
    which follows from our definition of $\be_{R}$. Thus, the
    covariance, $\cov(R_{i}X_{ij},A_{i3})$, reduces to the quadratic
    term, $\E[R_{i}X_{ij}A_{i3}]$, which can be expressed as 
    \begin{equation}\notag
      \begin{aligned}
        \E[R_{i}X_{ij}(\be_{R,i}-\be_{R})R_{i}]
        &= \E[R_{i}^{2}]\cdot\E[X_{ij}(\be_{R,i}-\be_{R})]\\
        &= \E[R_{i}^{2}]\cdot\E\big[\E[X_{ij}(\be_{R,i}-\be_{R})|X_{i}]\big]\\
        &= \E[R_{i}^{2}]\cdot\E\big[X_{ij}\E[(\be_{R,i}-\be_{R})|X_{i}]\big]\\
        &=0,
      \end{aligned}
    \end{equation}
    where the last equality follows from the first part of assumption (A2).
\end{proof}

\begin{proof}[Proof of Proposition \ref{pro:tsls consistency}]    
    We can first invoke proposition \ref{pro:uncorrelated}, which
    guarantees that $\cov(R_{i}X_{i},\ep_{i})=\bzero$, for every
    subject. Moreover, since the baseline variables, $X_{i}$'s, are
    assumed to be exogenous, it also follows that 
    $\cov(Z_{i},\ep_{i})=\bzero$, since we have defined 
    the $Z_{i}$'s as $(X_{i}\pri,R_{i},R_{i}X_{i}\pri)\pri$. Then, the proof
    of the consistency of $\atsls\bbe$ proceeds in a standard fashion by
    applying (A1). See chapter 5 of \citep{Wooldridge2002} for details. 
\end{proof}

\begin{proof}[Proof of Proposition \ref{pro:sps}]
   The optimal value of $\al$ is obtained after minimizing
   $f_{\al}:=\mse(\bP\bar\bbe_{\al})$. We will expand the trace of this criterion
   as was done in equation (\ref{eq:mse decomposition}), such that 
   \begin{equation}\notag
       \tr f_{\al} = \tr(\al^{2}M_{1} + 2\al(1-\al)C + (1-\al)^{2}M_{2}),
   \end{equation}
   with $M_{1}:=\mse(\bP\aols\bbe)$, $C:=\cse(\bP\aols\bbe,\bP\atsls\bbe)$, and
   $M_{2}:=\mse(\bP\atsls\bbe)$, respectively. Commuting the derivative
   operator with the trace, we obtain
   \begin{equation}\notag
        \tr(\partial f/\partial\al) = 
        2\al\tr(M_{2} - 2C + M_{1}) - 2\tr(M_{2} - C).
   \end{equation}
   Setting this expression to zero and solving for $\al$, yields $\al :=
   \tr(M_{2}-C)/\tr(M_{2}-2C + M_{1})$, as required. 

   In addition, a second derivative test can be performed in order to
   show that such minimizer is, in fact, a unique global minimizer.
   \begin{equation}\label{eq:second derivative}
        \tr(\partial^{2}f/\partial\al^{2}) = 2\tr(M_{1} - 2C + M_{2}).
   \end{equation}
   By assumption, the random vectors, $\aols\bbe$ and $\atsls\bbe$, are
   elementwise squared-integrable. Thus, the components, $\E[(\aols\be_{j}
   - \be_{j})^{2}]$, of $M_{1}$ are finite. Hence, using the linearity
   of the trace, the MSE of $\bP\aols\bbe$ can be treated as a sum of
   real numbers, thereby yielding the $L^{2}$-norm on $\R^{k+2}$,
   which we may denote by $||\bP(\aols\bbe-\bbe)||$. The latter quantity
   will be referred to as the (trace) RMSE of $\aols\bbe$.
   By the same reasoning, it can be shown that $C$ and
   $M_{2}$ corresponds to the inner product, $\lan\bP(\aols\bbe-\bbe),\bP(\atsls\bbe-\bbe)\ran$, and
   the squared norm, $||\bP(\atsls\bbe-\bbe)||^{2}$ on $\R^{k}$, respectively. Thus,
   equation (\ref{eq:second derivative}) can now be expressed as follows,
   \begin{equation}\notag
      \tr(\partial^{2}f/\partial\al^{2})
        = 2\Big(||\bP(\aols\bbe-\bbe)||^{2}
        - 2\lan\bP(\aols\bbe-\bbe),\bP(\atsls\bbe-\bbe)\ran
        + ||\bP(\atsls\bbe-\bbe)||^{2}\Big).
   \end{equation}
   The Cauchy-Schwarz inequality can then be used to produce an
   upper bound,
   \begin{equation}\notag
       \lan\bP(\aols\bbe-\bbe),\bP(\atsls\bbe-\bbe)\ran
       \leq ||\bP(\aols\bbe-\bbe)||\cdot||\bP(\atsls\bbe-\bbe)||. 
   \end{equation}
   Finally, by completing the square, we obtain the following lower bound,
   \begin{equation}\notag
      \tr(\partial^{2}f/\partial\al^{2})
      \geq 2\Big(||\bP(\aols\bbe-\bbe)|| - ||\bP(\atsls\bbe-\bbe)||\Big)^{2} \geq 0,
   \end{equation}
   for every $\bP$, and where equality solely holds when the RMSEs of
   the two estimators, $\aols\bbe$ and $\atsls\bbe$, are identical.
\end{proof}

\section{Appendix B: Simulation Model}
This appendix demonstrates how the variances of the error terms,
$\ep_{i}$'s and $\de_{i}$'s, denoted by $\sig^{2}_{\ep}$ and
$\sig^{2}_{\de}$ respectively; can be obtained in closed form, under
the constraints imposed upon our simulation model. 

We have here assumed the $X_{i}$'s to be exogenous, such that $X_{i}\perp
R_{i}$. Moreover, the confounders, denoted by $U_{i}$'s, have been assumed to
solely affect the relationship between the outcome and the mediator,
such that we also have $U_{i}\perp X_{i},R_{i},R_{i}X_{i}$.
Consequently, the variance of the $M_{i}$'s can be decomposed as
follows, 
\begin{equation}\label{eq:varM}
  \begin{aligned}
    \var(M_{i})\,=\,\, & \ga_{X}^{2}\var(X_{i}) + \ga_{R}^{2}\var(R_{i}) +
        \ga_{RX}^{2}\var(R_{i}X_{i}) + \ga_{U}^{2}\var(U_{i}) \\
      & + 2\ga_{X}\ga_{RX}\cov(X_{i},R_{i}X_{i}) 
        + 2\ga_{R}\ga_{RX}\cov(R_{i},R_{i}X_{i}) 
        + \sig^{2}_{\de};
  \end{aligned}
\end{equation}
using $X_{i}\perp R_{i}$, and the fact that the $U_{i}$'s are
independent of all the other variables on the RHS of equation
\eqref{eq:varM}. It can also easily be
seen that the mean and variance of the interaction variable
$R_{i}X_{i}$ are respectively given by $\E[R_{i}X_{i}]=\E[R_{i}]\E[X_{i}]=0$ and 
$\var(R_{i}X_{i})=\var(R_{i})\var(X_{i})=1$, by using the exogeneity of
the $X_{i}$'s, and the fact that the $X_{i}$'s are centered at zero. 
By a similar argument, the two covariances in equation \eqref{eq:varM}
can be simplified as follows, 
\begin{equation}\notag
    \cov(X_{i},R_{i}X_{i}) = \E[R_{i}]\E[X_{i}^{2}] = 1,
    \quad\te{and}\quad
    \cov(R_{i},R_{i}X_{i}) = \E[R_{i}^{2}]\E[X_{i}] = 0;
\end{equation}
since $R_{i}\in\{0,1\}$, and therefore $\E[R^{k}_{i}]=\E[R_{i}]$, for
every $k$. Hence, after fixing the variance of the $M_{i}$'s at $1$,
we can express the variance of the $\de_{i}$'s in terms of the
remaining parameters in that structural equation, such that
\begin{equation}\notag
    \sig^{2}_{\de}(\bga) = 
    1 - \big(2\ga_{X}^{2} + \frac{1}{4}\ga_{R}^{2} + \ga_{RX}^{2} + \ga_{U}^{2}
    + 2\ga_{RX}\ga_{X}\big),    
\end{equation}
with $\bga:=(\ga_{X},\ga_{R},\ga_{RX},\ga_{U})\pri$; and 
after using the Bernoulli distribution of the $R_{i}$'s, which gives
$\var(R_{i})=1/4$.
Throughout the simulations, the parameters controlling the effect of
the $X_{i}$'s and $R_{i}$'s have been set to $\ga_{X}:=1/4$, and
$\ga_{R}:=1/\sqrt{2}$, respectively. This choice of parameters has
been selected in order to simplify the expression for
$\sig^{2}_{\de}$, such that we obtain, $\sig^{2}_{\de}=0.75 -
\ga_{RX}^{2} - \frac{1}{2}\ga_{RX} -\ga_{U}$. 

Similarly, we can standardize the variance of the outcome variables,
$Y_{i}$'s. Given that the $X_{i}$'s, $R_{i}$'s and $M_{i}$'s are
cross-correlated, this produces a convoluted formula given by the following,
\begin{equation}\label{eq:varY}
  \begin{aligned}
    \var(Y_{i})\,=\,\, & \be_{X}^{2}\var(X_{i}) + \be_{R}^{2}\var(R_{i}) +
        \be_{M}^{2}\var(M_{i}) + \be_{U}^{2}\var(U_{i}) \\
       & + 2\be_{X}\be_{M}\cov(X_{i},M_{i}) 
        + 2\be_{R}\be_{M}\cov(R_{i},M_{i}) \\
       & + 2\be_{M}\be_{U}\cov(M_{i},U_{i}) + \sig^{2}_{\ep};
  \end{aligned}
\end{equation}
after applying $R_{i}\perp X_{i}$, and using the fact that the $U_{i}$'s are
independent of both the $R_{i}$'s, and the $X_{i}$'s. Equation
\eqref{eq:varY} can be further simplified by using our choice of
parametrization, which gives $\cov(X_{i},M_{i})=2\ga_{X}+\ga_{RX}$, 
$\cov(R_{i},M_{i})=\ga_{R}/4$, and
$\cov(U_{i},M_{i})=\ga_{U}$. Altogether, we therefore obtain a
closed-form formula for the variance of the error terms of the
$Y_{i}$'s, expressed in terms of the model parameters,
$\bbe:=(\be_{X},\be_{R},\be_{M},\be_{U})$ and $\bga$. That is, 
\begin{equation}\notag
  \begin{aligned}
    \sig^{2}_{\ep}(\bbe,\bga) = 1 - 
    \big(2\be_{X}^2 + \frac{1}{4}\be_{R}^2 + \be_{M}^2 + \be_{U}^2 + C\big),
  \end{aligned}
\end{equation}
where $C:=2\be_{X}\be_{M}(2\ga_{X}+\ga_{RX}) + \be_{R}\be_{M}\ga_{R}/2
+ 2\be_{M}\be_{U}\ga_{U}$. Therefore, we have obtained closed form
formulas for both $\sig^{2}_{\de}$ and $\sig^{2}_{\ep}$. These
formulas have then be used to constrain the range of the parameters of
interest, in the different simulation scenarios.

\begin{acks}
This work was supported by an 
MRC project grant MR/K006185/1, Landau et al.~(2013-2016) entitled
``Developing methods for understanding mechanism in complex
interventions.'' This research is partly funded by the National
Institute for Health Research (NIHR) Biomedical Research Centre at
South London and Maudsley NHS Foundation Trust and King's College
London. The views expressed are those of the authors and not
necessarily those of the NHS, the NIHR or the Department of Health.
RE was supported by the MRC North West Hub for Trials Methodology
Research (MR/K025635/1).
\end{acks}

\bibliographystyle{SageV} 
\bibliography{/home/cgineste/ref/bibtex/Statistics,%
             /home/cgineste/ref/bibtex/Neuroscience}
\end{document}